\documentclass[11pt]{article}
\usepackage{fullpage}
\usepackage[latin1]{inputenc}
\usepackage[hypertexnames=false]{hyperref}
\usepackage{graphicx,cite}
\usepackage{epsfig}
\usepackage{color}
\usepackage{enumerate,dsfont}
\usepackage{verbatim}
\usepackage{amsthm, amssymb}
\usepackage{amsmath}

\newtheorem{lemma}{Lemma}
\newtheorem{theorem}{Theorem}
\newtheorem{definition}{Definition}
\newtheorem{corollary}{Corollary}
\newtheorem{invariant}{Invariant}

\newcommand{\intDDG}{\ensuremath{\mathit{int}\mathrm{DDG}}}
\newcommand{\extDDG}{\ensuremath{\mathit{ext}\mathrm{DDG}}}
\newcommand{\into}[1]{\ensuremath{\mathit{int}(#1)}}
\newcommand{\exto}[1]{\ensuremath{\mathit{ext}(#1)}}
\newcommand{\intc}[1]{\ensuremath{\overline{\mathit{int}}(#1)}}
\newcommand{\extc}[1]{\ensuremath{\overline{\mathit{ext}}(#1)}}
\newcommand{\lca}[2]{\ensuremath{\mathit{lca}(#1,#2)}}
\newcommand{\jump}{\ensuremath{\mathit{jump}}}

\title{Min $st$-Cut Oracle for Planar Graphs with Near-Linear Preprocessing Time}
\author{Glencora Borradaile\thanks{Oregon State University} \and
      Piotr Sankowski\thanks{University of Warsaw}
      \and
        Christian Wulff-Nilsen
        \thanks{University of Copenhagen                 }}

\begin{document}

\maketitle

\begin{abstract}
  For an undirected $n$-vertex planar graph $G$ with non-negative edge-weights,
  we consider the following type of query: given two vertices $s$ and
  $t$ in $G$, what is the weight of a min $st$-cut in $G$? We show how
  to answer such queries in constant time with $O(n\log^4n)$
  preprocessing time and $O(n\log n)$ space.  We use a Gomory-Hu tree
  to represent all the pairwise min cuts implicitly. Previously,
  no subquadratic time algorithm was known for this problem.  Since
  all-pairs min cut and the minimum cycle basis are dual problems in
  planar graphs, we also obtain an implicit representation of a
  minimum cycle basis in $O(n\log^4n)$ time and $O(n\log n)$ space. Additionally, an
  explicit representation can be obtained in $O(C)$ time and space where
  $C$ is the size of the basis.

  These results require that shortest paths are unique. This can be
  guaranteed either by using randomization without overhead, or
  deterministically with an additional $\log^2 n$ factor in the
  preprocessing times.
\end{abstract}

\section{Introduction}

A minimum cycle basis is a minimum-cost representation of all the
cycles of a graph and the all-pairs min cut problem asks to find all
the minimum cuts in a graph.  In planar graphs the problems are
intimately related (in fact, equivalent~\cite{HM94}) via planar
duality. In this paper, we give the first sub-quadratic algorithm for these problems; it
runs in $O(n \log^4 n)$ time. This result is randomized, but can be
made deterministic by paying an additional $\log^2n$ factor in the
running time. In the following, we consider connected, undirected graphs with
non-negative edge weights.

\subsubsection*{All-pairs minimum cut}

In the {\em all-pairs min cut problem} we need to find the minimum $st$-cut
for every pair $\{s,t\}$ of vertices in a graph $G$. Gomory and Hu~\cite{GH61}
showed that these minimum cuts can be represented by an edge-weighted tree such that:
\begin{itemize}
\item the nodes of the tree correspond one-to-one with the vertices of $G$,
\item for any distinct vertices $s$ and $t$, the minimum-weight edge
  on the unique $s$-to-$t$ path in the tree has weight equal to the
  min $st$-cut in $G$, and
\item removing this minimum-weight edge from the tree
  creates a partition of
  the nodes into two sets
  corresponding to a min $st$-cut in $G$.
\end{itemize}
We call such a tree a {\em Gomory-Hu} tree or GH tree -- it is also referred
to as cut-equivalent and cut tree in the literature. Gomory and Hu
 showed how to find such a tree with $n-1$ calls to a minimum cut
algorithm.  Up to date, this is the best known method for general graphs
and gives an $O(n^2 \log \log n)$-time algorithm for planar graphs
using the best-known algorithm for min $st$-cuts in undirected planar
graphs~\cite{INSW11}. There exists an algorithm for unweighted,
general graphs that beats the $n-1$ times minimum cut time
bound~\cite{BHPT07}; the corresponding time for unweighted planar graphs is,
however, $O(n^2 \mathrm{poly}\log n)$ time.

\subsubsection*{Minimum cycle basis}

A cycle basis of a graph is a maximum set of independent cycles.  Viewing a
cycle as an incidence vector in $\{0,1\}^E$, a set of cycles is
independent if their vectors are independent over $GF(2)$.  The weight
of a set of cycles is the sum of the weights of the cycles.  The {\em
  minimum-cycle basis (MCB) problem} is to find a cycle basis of
minimum weight.  This problem dates to the electrical circuit theory
of Kirchhoff~\cite{Kirchhoff1847} in $1847$ and has been used in the
analysis of algorithms by Knuth~\cite{Knuth68}.  For a complete survey,
see \citen{Horton87}.  The best known algorithm in general graphs takes
$O(m^\omega)$ time where $\omega$ is the exponent for matrix
multiplication~\cite{AIJKR09}.

The best MCB algorithms for planar graphs use basic facts of planar
embeddings. Hartvigsen and Mardon~\cite{HM94} prove that if $G$ is
planar, then there is a minimum cycle basis whose cycles are simple
and nested in the drawing in the embedding.  ({\em Nesting} is defined
formally in Section~\ref{sec:iso}.) As such, one can represent a
minimum cycle basis of a planar embedded graph as an edge-weighted
tree, called the {\it MCB tree}, such that:
 \begin{itemize}
 \item the nodes of the tree correspond one-to-one with the faces of
   the planar embedded graph, and
 \item each edge in the tree corresponds to a cycle in the basis,
   namely the cycle that separates the faces in the components
   resulting from removing this edge from the tree.
 \end{itemize}
Hartvigsen and Mardon also gave an $O(n^2 \log n)$-time algorithm for
the problem that was later improved to $O(n^2)$ by Amaldi et~al.~\cite{AIJKR09}.

\subsubsection*{Equivalence between MCB and GH trees}

In planar graphs, the MCB and GH problems are related via planar
duality.  Corresponding to every connected planar embedded graph $G$
(the {\em primal}) there is another connected planar embedded graph
(the {\em dual}) denoted $G^*$. The faces of $G$ are the vertices of
$G^*$ and vice versa. There is a one-to-one correspondence between the
edges of $G$ and the edges of $G^*$: for each edge $e$ in $G$, there
is an edge $e^*$ in $G^*$ whose endpoints correspond to the faces of $G$
incident to $e$.  Dual edges inherit the weight of the corresponding primal edge; namely, $w(e^*) = w(e)$.

We define {\it a simple cut} to be a cut such that both sides of the
cut are connected. This definition allows us to show that cycles and
cuts are equivalent through duality:
\begin{quote}
  In a connected planar graph, a set of edges forms a cycle in the
  primal iff it forms a simple cut in the dual.~\cite{Whitney1933}
\end{quote}
Just as cuts and cycles are intimately related via planar duality, so
are the all-pairs minimum cut and minimum cycle basis problems.  In fact,
Hartvigsen and Mardon showed that they are equivalent in the following sense:
\begin{theorem}[Corollary 2.2 from \citen{HM94}]\label{thm:cycle-basis-dual-GH}
  For a planar embedded graph $G$, a tree $T$ represents a minimum
  cycle basis of $G$ if and only if $T$ is a Gomory-Hu tree for $G^*$
  (after mapping edges to their duals and the node-face relationship
  to a node-vertex relationship via planar duality).
\end{theorem}
\noindent It follows that the $O(n^2)$ algorithm due to Amaldi
et~al.~\cite{AIJKR09} to find an MCB tree is also an algorithm to find
a GH tree.

\paragraph{Techniques} Herein, we focus on the frame of reference of the
minimum cycle basis.  In the remainder of this section, we highlight
the tools and techniques that we use in our algorithm and analysis.
Our algorithm, at a very high level, works by (iterative) finding a
minimum-weight cycle that separates two as-yet unseparated faces $f$
and $g$.  The order in which we separate pairs of faces is guided by a
bottom-up traversal of a hierarchical planar separator
(Section~\ref{sec:sepintro}).  We find each separating cycle by a
modification of Reif's algorithm (Section~\ref{sec:reifintro}) which
does so by way of shortest-path computations.  In order to achieve a
sub-quadratic running time, we will precompute many of the distances
for these shortest-path computations using a method developed by
Fakcharonphol and Rao (Section~\ref{sec:DenseDistGraph}).  To
guarantee that the cycles we find, using these precomputed distances,
will be nesting, we rely on shortest paths being unique
(Section~\ref{sec:iso}).  We represent a partially built solution with
something we call a {\em region tree} (Section~\ref{sec:region-tree}).
After we have defined these structures and tools, we will give a more
detailed overview of our algorithm and analysis
(Section~\ref{sec:overview}) before diving into the technical details
in the remainder of the paper.

\subsection{Planar graphs and simplifying assumptions} \label{sec:defs}

An embedded planar graph is a mapping of the vertices to distinct
points and edges to non-crossing curves in the plane.  A face of the
embedded planar graph is a maximal open connected set of points that
are not in the image of any embedded edge or vertex. Exactly one face,
the {\em infinite face} is unbounded. We identify a face with the
embedded vertices and edges on its boundary.

For a simple cycle $C$ in a planar embedded graph $G$, let $\into{C}$
denote the open bounded subset of the plane defined by $C$.  Likewise
define $\exto{C}$ for corresponding unbounded subset. We refer to the
closure of these sets as $\intc{C}$ and $\extc{C}$, respectively.  We
say that a pair of faces of $G$ are \emph{separated} by $C$ in $G$ if
one face is contained in $\intc{C}$ and the other face is contained in
$\extc{C}$.  A set of simple cycles of $G$ is called \emph{nested} if
mapping of this set $\into{}$ is also nesting.  A simple cycle $C$ is
said to \emph{cross} another simple cycle $C'$ if the cycles are not
nested.

\subsubsection*{Unique shortest paths} Our algorithm relies on the fact
that each new cycle we add to the basis nests with previously added
cycles.  To guarantee this (Section~\ref{sec:iso}), we will assume
that in all the graphs we consider ($G$ and graphs obtained from $G$),
there is a unique shortest path between any pair of vertices.  By
adding a small, random perturbation to the weight of each edge, one
can make the probability of having non-unique shortest paths
arbitrarily small. For example, if we take a random perturbation from
the set $[1,\ldots, n^4] \frac{1}{n^7}$, then by the Isolation Lemma
due to Mulmuley et al.~\cite{MVV87} the probability that all shortest
paths are unique is at least $1 - \frac{1}{n^3} \times {n \choose 2}
\ge 1 - \frac{1}{n}$.  

In Section~\ref{sec:unique-short-paths}, we give a more robust,
deterministic way to ensure the uniqueness of shortest paths.  We must
impose the structural simplifications presented below before
considering applying this more robust method.  The idea, based on the
technique used by Hartvigsen and Mardon~\cite{HM94}, is to break ties
consistently, thus imposing uniqueness on the shortest paths so the
above two lemmas hold.  Unfortunately, in order to do so, we require a
$\log^2 n$-increase in the running time of our algorithm.

\subsubsection*{3-regular with small, simple faces} The separators that we
use will require that the boundaries of the faces be small and simple
(each vertex appears only once on the boundary of the face).
Three-regularity (each vertex having degree 3) greatly simplifies the
analysis of our algorithm.  We can modify the input graph to satisfy
both these properties simultaneously by triangulating the primal with
infinite-weight edges and then triangulating the dual with zero-weight
edges using a {\em zig-zag triangulation}.  In a zig-zag
triangulation~\cite{KB92}, each face is triangulated using a simple
path, adding at most two edges per face adjacent to each vertex.

\paragraph{Triangulating the primal} For each face $f$ of the original
graph, identify a face $f'$ of the triangulated graph that is enclosed
by $f$ in the inherited embedding.  The minimum $fg$-separating cycle
in the original graph maps to an $f'g'$-separating cycle of the same
weight and so will not use any infinite-weight edge: the set of cycles
in a minimum cycle basis in the original graph are mapped to the set
of finite-weight cycles in a minimum-cycle basis of the finite-weight
cycles of the triangulated graph.

\paragraph{Triangulating the dual} Before the zig-zag triangulation,
each vertex in the dual has degree 3.  The zig-zag triangulation adds
at most 6 edges (2 for every adjacent face) to each vertex.
Therefore, the faces in the primal have size at most 9 and are still
simple.  Clearly, the primal is 3-regular after triangulating the
dual.  As in the triangulation of the primal, we can map between
minimum cycle bases and min cuts in the original graph and the
degree-three graph. Each vertex $v$ in the original graph is mapped to
a path $P_v$ of zero-weight edges in the degree-three graph.

\subsection{Reif's algorithm for minimum separating cycles} \label{sec:reifintro}

Reif gave an algorithm for finding minimum cuts by way of finding
minimum separating cycles in the dual graph~\cite{Reif83}.  

Let $X$ be the shortest path between any vertex on the boundary of $f$
and any vertex on the boundary of $g$.  Since $X$ is a shortest path,
there is a minimum $fg$-separating cycle, $C$, that crosses $X$ only
once.  Paths $P$ and $Q$ cross if there is a quadruple of faces
adjacent to $P$ and $Q$ that cover the set product $\{\mbox{left of
  $P$},\mbox{right of $P$}\}\times\{\mbox{left of $Q$},\mbox{right of
  $Q$}\}$.  Let $G_X$ be the graph obtained from $G$ by {\em cutting
  along} path $X$: duplicate every edge of $X$ and every internal
vertex of $X$ and create a new, simple face whose boundary is composed
of edges of $X$ and their duplicates.

The following result is originally due to Itai and
Shiloach~\cite{IS79} but we state it as it was given by
Reif~\cite{Reif83}.
\begin{theorem}[Proposition~3 from~\cite{Reif83}]\label{thm:cross-cycle}
  Let $X$ be the shortest $f$-to-$g$ path.  For each vertex $x \in X$,
  let $C_x$ be the minimum weight cycle that crosses $X$ exactly once
  and does so at $x$.  Then $\min_{x\in X} C_x$ is a minimum
  $fg$-separating.  Further, $C_x$ is the shortest path between
  duplicates of $x$ in $G_X$.
\end{theorem}
This theorem is algorithmic: the shortest paths between duplicates of
vertices on $X$ in $G_X$ can be found in $O(n \log n)$ time using
Klein's multiple-source shortest path algorithm~\cite{Klein05} or by
using the linear-time shortest-path algorithm for planar
graphs~\cite{HKRS97} and divide and conquer.  In our algorithm we will
emulate the latter method: start with the midpoint, $x$, of $X$ in
terms of the number of vertices, and recurse on the subgraphs obtained
by cutting along $C_x$.

\subsection{Isometric cycles} \label{sec:iso}


A cycle $C$ in a graph is said to be \emph{isometric} if for
any two vertices $u,v\in C$, there is a shortest path in the graph
between $u$ and $v$ which is a subpath of $C$. A set of cycles is said
to be isometric if all cycles in the set are isometric. 
\begin{lemma}[Proposition 4.4 from \citen{HM94}]
Any minimum cycle basis of a graph is isometric.
\end{lemma}

The following lemma will allow us to find isometric cycles by
composing shortest paths. Further, these isometric cycles will be
nesting and so can be represented with a tree which is precisely the
MCB tree.
\begin{lemma}\label{Lem:Isometric}
  Let $G$ be a graph in which shortest paths are unique.  The
  intersection between an isometric cycle and a shortest path in $G$
  is a (possibly empty) shortest path. The intersection between two
  distinct isometric cycles $C$ and $C'$ in $G$ is a (possibly empty)
  shortest path; in particular, if $G$ is a planar embedded graph, $C$
  and $C'$ do not cross.
\end{lemma}
\begin{proof}
  Let $C$ be an isometric cycle and let $P$ be a shortest path
  intersecting $C$. Let $u$ and $v$ be the first and last vertices of
  $P$ that are in $C$. Since $C$ is isometric, there is a shortest
  path $P'$ in $C$ between $u$ and $v$. Since shortest paths are
  unique, $P'$ is the subpath of $P$ between $u$ and $v$.  Hence,
  $C\cap P = P'$, giving the first part of the lemma.

  Let $C'$ and $C$ be distinct isometric cycles.  For any two distinct
  vertices $u,v \in C' \cap C$, let $P$ be the shortest $u$-to-$v$
  path.  Since $C$ and $C'$ are isometric and shortest paths are
  unique, $P \subset C' \cap C$.
\end{proof}

\subsection{Planar separators} \label{sec:sepintro}

A \emph{decomposition} of a graph $G$ is a set of subgraphs
$P_1,\ldots,P_k$ such that the union of vertex sets of these subgraphs
is the vertex set of $G$ and such that every edge of $G$ is contained
in a unique subgraph. We call $P_1,\ldots,P_k$ the \emph{pieces} of
the decomposition. The \emph{boundary vertices} $\partial P_i$ of a
piece $P_i$ is the set of vertices $u$ in that piece such that there
exists an edge $(u,v)$ in $G$ with $v\notin P_i$. We recursively
decompose the graph to get a {\em recursive subdivision}.  A piece $P$
is decomposed into subpieces called the children of $P$; the boundary
vertices of a child are the boundary vertices inherited from $P$ as
well as the boundary vertices introduced by the decomposition of $P$.
We use Miller's Cycle Separator for this decomposition: the introduced
boundary vertices form simple cycles~\cite{Miller86}.  This
decomposition requires that the sizes of the faces be bounded by a constant.

Fixing an embedding of $G$, a piece inherits its embedding from $G$'s
embedding. A {\it hole} in a piece is a bounded face containing
boundary vertices.  While it is not possible to guarantee that holes
are not introduced by the balanced recursive subdivision, it is
possible to guarantee that each piece has a constant number of
holes~\cite{INSW11,KMS13}.  Further, we will ensure that pieces are
connected; we give the details of ensuring connectivity in
Section~\ref{sec:sep}.

In summary, we use the following recursive decomposition:
\begin{definition}[Balanced recursive subdivision] \label{def:sep} A
  decomposition of $G$ such that a piece $P$ is divided into a
  constant number of subpieces, $P_1, P_2, P_3, \ldots$ each of which
  is connected, has $O(1)$ holes, and contains at most $\frac 1 2 |P|$
  vertices, at most $\frac 1 2 |\partial P|$ boundary vertices
  inherited from $P$, and at most $\sqrt {|P|}$ additional boundary
  vertices.
\end{definition}

We define the $O(\log n)$ \emph{levels} of the recursive subdivision
in the natural way: level $0$ consists of one piece ($G$) and level
$i$-pieces are obtained by applying the Cycle Separator Theorem to
each level $(i-1)$-piece. We represent the recursive subdivision as a
binary tree, called the \emph{subdivision tree (of $G$)}, with level
$i$-pieces corresponding to vertices at level $i$ in the subdivision
tree. Parent/child and ancestor/descendant relationships between
pieces correspond to their relationships in the subdivision tree.  For
notation simplicity, we assume that the subdivision tree is binary;
generalizing this to a constant number of children is straightforward.

We prove the following in Section~\ref{sec:sep}:
\begin{theorem}\label{thm:sumPiece}
  Let $\mathcal P$ be the set of pieces in a recursive subdivision of
  $G$.  Then $\sum_{P\in\mathcal P}|P| = O(n\log n)$ and
  $\sum_{P\in\mathcal P}|\partial P|^2 = O(n\log n)$.
\end{theorem}

\subsection{Precomputing distances}\label{sec:DenseDistGraph}

For a piece $P$, the \emph{internal dense distance graph of $P$} or
$\intDDG(P)$ is the complete graph on the set of boundary vertices of
$P$, where the weight of each edge $(u,v)$ is equal to the shortest
path distance between $u$ and $v$ in $P$. The union of internal dense
distance graphs of all pieces in the recursive subdivision of $G$ is
the \emph{internal dense distance graph (of $G$)}, or simply
$\intDDG$.  Fakcharoenphol and Rao showed how to compute $\intDDG$ in
$O(n\log^3n)$ time~\cite{FR06}.  (Using the multiple-source shortest
path algorithm due to Klein, this can be improved to
$O(n\log^2n)$~\cite{Klein05}; however, this improvement is not
compatible with, deterministically imposing unique
shortest paths or generalizing to external distances, so we do not use it.)

Consider a piece $P$ with $h$ holes.  The \emph{external dense
  distance graph of $P$} or $\extDDG(P)$ is the union of $h+1$
complete graphs: one for each hole and one for the external face.  For
two vertices $u,v$ on the boundary of a common hole or on the boundary
of the external face, the weight of edge $(u,v)$ is the shortest path
distance between $u$ and $v$ in the component of $G\setminus E(P)$
that contains $u$ and $v$.  The \emph{external dense distance graph of
  $G$} or $\extDDG$ is the union of all external dense distance graphs
of the pieces in the recursive subdivision of $G$.  As a consequence
of Theorem~\ref{thm:sumPiece}, both $\intDDG$ and $\extDDG$ have size
$O(n\log n)$.

Fakcharoenphol and Rao give an implementation of Dijkstra's algorithm,
{\em FR-Dijkstra}, that finds a shortest path tree in a graph composed
of dense distance graphs in $O(b\log^2b)$ time, where $b$ is the total
number of boundary vertices, counted with multiplicity~\cite{FR06}.

\begin{theorem}\label{thm:ExtDenseDistGraph}
    The external dense distance graph of 
  $G$ can be computed in $O(n\log^3n)$ time.
\end{theorem}
\begin{proof}
  Fakcharoenphol and Rao compute $\intDDG$ via a leaves-to-root
  traversal of the recursive subdivision of $G$ by applying
  FR-Dijkstra to obtain $\intDDG(P)$ for a piece $P$ from the internal
  dense distance graphs of its two children.  Consider the external
  distances for the external face of $P$: these can be computed by
  FR-Dijkstra from the external dense distance graph of the parent of
  $P$ and the internal dense distance graph of the sibling(s) of
  $P$.  

  Consider a hole $H$ of a child $P'$ of $P$ and let us compute
  $\extDDG(P')$ restricted $H$. Observe that (the subgraph of $G$
  restricted to) $H$ is the union of certain sibling pieces of $P'$
  and certain holes of $P$. Given internal dense distance graphs of
  all children of $P$ and external dense distance graphs of all holes
  of $P$, it follows that $\extDDG(P')$ restricted to $H$ can be
  obtained efficiently with FR-Dijkstra. 

  Therefore, we can obtain $\extDDG$ with a root-to-leaves traversal
  of the recursive subdivision given $\intDDG$. The running time is
  the same as that for finding $\intDDG$.
\end{proof}

\subsection{The region tree} \label{sec:region-tree}

We build the minimum cycle basis iteratively, adding cycles to a
partially constructed basis that separated two as-yet unseparated
faces.  We represent the partially constructed basis with a tree that
we call the {\em region tree}.  Each node of a region tree either
represents a {\em region} (defined below) or a face of the graph.  Adjacency
represents enclosure.  We take the region tree to be rooted at a root
$r$ which always corresponds to the special region that represents the
entire plane.  Initially the tree is a star centered at a root $r$
with each leaf corresponding to a face in the graph (including the
infinite face).

\begin{figure}[ht]
  \centering
  \includegraphics[scale=0.7]{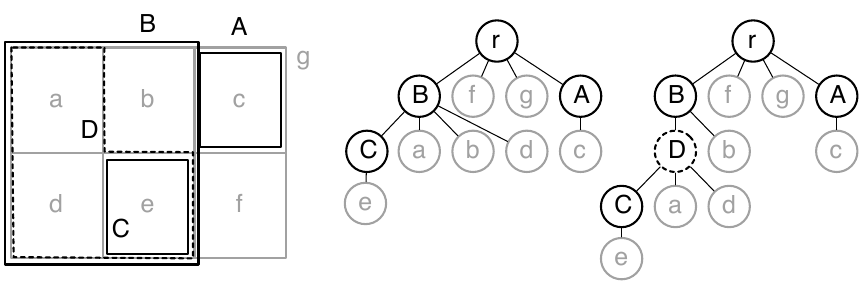}
  \caption{A graph with faces $a$ through $g$; four nesting cycles
    $A$ through $D$ (left).  A region tree for cycles
    $A$,$B$ and $C$ (center).  A region tree for cycles $A$
    through $D$ (right).  }
  \label{fig:region-tree-update}
\end{figure}

We update the tree to reflect the cycles that we add to the basis.
This process is illustrated in Figure~\ref{fig:region-tree-update}.
When a first cycle $C$ is found, we create a new node $x_C$ for the
tree, make $x_C$ a child of the root, and make all the faces that $C$
encloses children of $x_C$.  Node $x_C$ defines a \emph{region} $R$:
the subgraph of $G$ contained in the closed subset of the plane
defined by the interior of $C$ and not enclosed in the interior of any
children of $C$. We say that $R$ is {\em bounded} by $C$, that $C$ is
a \emph{bounding cycle} of $R$, and that $R$ {\em contains} the child
regions and/or child faces defined by the tree.  If two faces $f,g$
are unseparated, they are children of a common region node $r$; to add
a cycle $C$ that separates $f$ and $g$, $x_C$ will become a child of
$f$ and $g$'s common parent.  As we will only add cycles which nest
with those we have already found, the updates to the region tree are
well defined.  The most technically challenging part of the algorithm
is in how to update the region tree (Section~\ref{sec:UpdateRegionTree}).

In the final tree, after all faces have been separated, each face is
the only face-child of a region.  We call such a region tree a {\em
  complete region tree}.  Mapping each face to its unique parent
creates a tree with one node for each face in the graph; this is the
MCB-tree.

\paragraph{Region tree data structure} We represent the region tree
using the top tree data structure~\cite{AHLT05}. This will allow us to
find the lowest common ancestor $\lca{x}{y}$ of two vertices $x$ and
$y$, and determine whether one vertex is a descendant of another in
logarithmic time.  Top trees also support the operation
$\jump(x,y,d)$, which for two vertices $x$ and $y$, finds the vertex
that is $d$ edges along the path from $x$ to $y$ in logarithmic time.
Top trees can also find the weight of the simple path between two
given endpoints in logarithmic time.

\subsection{Region subpieces}

The \emph{region subpieces} of a piece $P$ are the subgraphs defined
by the non-empty intersections between $P$ and regions defined by the
region tree. We say that a region $R$ and a region subpiece $P_R$ are
associated with each other.  The boundary vertices $\partial P_R$ of
$P_R$ are inherited from $P$: $\partial P_R = P_R \cap \partial P$.
These constructions are illustrated in Figure~\ref{fig:region-subpiece}.

\begin{figure}[ht]
  \centering
  \includegraphics[scale=0.5]{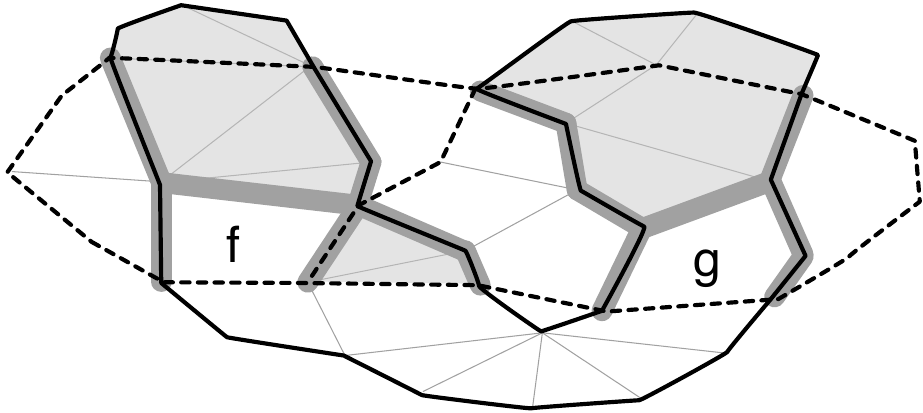}
  \caption{The dotted edges are the edges belonging to the boundaries
    of a piece $P$ and $P$'s children $P_1$ and $P_2$.  $f$ is a face
    of $P_1$ and $g$ is a face of $P_2$.  The solid black edges are
    the bounding cycle of a region $R$.  The three shaded regions are
    three child regions of $R$.  The thick grey edges are the edges of
    region subpiece $P_R$. (The remaining edges of the graph are the
    thin, grey edges.)  The intersection of a minimum separating cycle
    for $f$ and $g$ with $P$ uses only edges of $P_R$.}
  \label{fig:region-subpiece}
\end{figure}

The order in which we separate faces is guided in a bottom-up fashion
by the recursive subdivision of $G$.  Starting with a piece $P$ at the
deepest level of the recursive subdivision, we separate all pairs of
faces of $G$ that have an edge in common with $P$, allowing us to
maintain the following invariant:

\begin{invariant}\label{inv:one-pair}
  For each region subpiece $P_R$ of $P$, at most one pair of faces
  needs to be separated.
\end{invariant}

\begin{proof}
  Suppose we have added (nesting) separating cycles to the region tree
  that separate all pairs of faces sharing an edge with a descendant
  piece of $P$. Suppose two faces $f_1$ and $f_2$ of $G$ both sharing
  edges with $P$ have not yet been separated. Let $P_1$ and $P_2$ be
  the two child pieces of $P$ in the recursive subdivision. Since all
  pairs of faces sharing edges with $P_1$ and all pairs of elementary
  faces sharing edges with $P_2$ have already been separated,
  w.l.o.g.~$f_i$ only shares edges with $P_i$.

  Since $f_1$ and $f_2$ have not been separated, they must belong to a
  common region $R$ in the region tree. For a contradiction, if there
  is a third face $f$ that is unseparated from $f_1$ or $f_2$, $f$
  w.l.o.g.~shares edges with $P_1$. However, all pairs of faces
  sharing edges with $P_1$ have already been separated.
\end{proof}

\subsection{Overview of the algorithm and analysis}\label{sec:overview}

Our algorithm for computing an MCB tree is:
\begin{quote}
Find a balanced recursive decomposition of the input graph.

Compute the internal and external dense distance graphs. 

Initialize the region tree that represents an empty set of cycles.

Considering each piece $P$ of the recursive decomposition, according
to a bottom-up order, determine the region subpieces of $P$.

For each region subpiece $P_R$ that contains a pair of unseparated
faces:
\begin{itemize}
\item Find the minimum $fg$-separating cycle $C$.
\item Add $C$ to the collection of cycles and update the region tree accordingly.
\end{itemize}
\end{quote}

We present the remaining details for the balanced, recursive
decomposition (Definition~\ref{def:sep}), namely, ensuring that the
pieces remain connected and the proof of Theorem~\ref{thm:sumPiece},
in Section~\ref{sec:sep}.  We have already discussed computing the
dense distance graphs (Section~\ref{sec:DenseDistGraph}) and
initializing the region tree (Section~\ref{sec:region-tree}).

\bigskip

In Section~\ref{sec:RegionSubpieces}, we show, given a piece
$P$ and the current set of regions represented by the region tree,
what the region subpieces for $P$ are and which edges of the graph are
in each of the region subpieces.  Formally, we will prove:
\begin{theorem}\label{thm:RegionSubpiece}
  The region subpieces of a piece $P$ can be identified in $O((|P| +
  |B|^2)\log^2n)$ time, where $B$ is the set of boundary vertices of
  the children of the piece.
\end{theorem}
\noindent Summing over all region
subpieces, appealing to Theorem~\ref{thm:sumPiece}, the time spent by
the algorithm identifying the region subpieces is $O(n \log^3 n)$.

\bigskip

We have already seen that by considering the pieces in bottom-up
order, each region subpiece has at most one pair $(f,g)$ of
unseparated faces (Invariant~\ref{inv:one-pair}).  In
Section~\ref{sec:SepFacePair}, we show how to find a minimum
$fg$-separating cycle $C$ by emulating Reif's algorithm.  To do so, we
must implicitly cut open the graph along the shortest $f$-to-$g$ path
and modify the dense distances to reflect the change in the graph's
shortest path metric.  We will prove:
\begin{theorem}\label{thm:separate}
  The minimum $fg$-separating cycle for a region subpiece $P_R$ can be found in time $O((|\partial P_R|^2+|P_R|)\log^3|P_R|)$.  
\end{theorem}
\noindent Therefore, the time spent by the algorithm finding minimum separating cycles is $O(n
\log^4 n)$.

\bigskip

We show how to update the region tree to reflect our addition of cycle
$C$ to the basis.  For each region $R$ in the region tree we store a
{\em compact representation} $G[R]$ where vertices with degree 2 are
removed by merging the adjacent edges creating {\em super edges}. For
each super edge we store the first and the last edge on the
corresponding path. We show how to maintain compact representations of
regions in Section~\ref{sec:UpdateRegionTree}.  Formally, we will prove:
\begin{theorem}\label{thm:updateregiontree}
  We can update the region tree to reflect that region $R$ is split
  into regions $R_1$ and $R_2$ by the addition of $C$ in time
  $O(\min\{|F_1|,|F_2|\}\log^3 n+ (|P_R|+|\partial P_R|^2)\log n)$
  where $F_i$ are the children of $R_i$ in the region tree after the
  update.
\end{theorem}
\noindent Using the lemma below, the total time spent by the algorithm
in updating the region tree is $O(n \log^4 n)$.  Combining with the
above running times, this gives an overall running time of $O(n \log^4
n)$.

\begin{lemma}\label{Lem:Merge}
  Consider a set of objects $S$, a weight function $w: S \rightarrow
  \mathds{Z}^+$ and a merging operation that replaces distinct objects
  $o$ and $o'$ by a new object whose weight is at most $w(o) + w(o')$
  in time at most $t\cdot\min\{w(o),w(o')\}$ for some $t$. Then
  repeatedly merging objects until one object remains takes $O(t \,
  w(S)\log w(S))$ time.
\end{lemma}
\begin{proof}
  We may suppose w.l.o.g.\ that initially all objects have weight $1$.
  The run-time for a sequence of merges that results in an object of
  weight $w$ satisfies the recurrence
  \[
    T(w)\leq c\, \max_{1\leq w'\leq\lfloor w/2\rfloor}\{T(w') + T(w - w') + tw'\}
  \]
  for some constant c.  It is easy to see that the right-hand side is
  maximized when $w' = \lfloor w/2\rfloor$, giving $T(w(S)) =
  O(t\,w(S)\log w(S))$, as desired.
\end{proof}

\subsubsection*{Results} Recall that these running times are stated with
the uniqueness-of-shortest-paths assumption (guaranteed by suitable
randomization) and that we will show how to achieve this uniqueness
deterministically while incurring an additional $O(\log^2 n)$ factor
in the running times.  Our algorithm computes the complete region tree
in $O(n \log^4 n)$ time.  As argued in Section~\ref{sec:region-tree},
this can be used to obtain the MCB tree; by planar duality, the same algorithm can be used to find the GH tree.  Therefore, we get: 

\begin{theorem}
  The minimum cycle basis or Gomory-Hu tree of an undirected and
  unweighted planar graph can be computed in $O(n\log^4n)$ time and
  $O(n\log n)$ space.
\end{theorem}

In order to find a minimum $st$-cut using the GH tree, we need to find
the minimum weight edge on the $s$-to-$t$ path in the tree.  With an
additional $O(n \log n)$ preprocessing time, one can answer such
queries in $O(1)$ time using a tree-product data
structure~\cite{KS08}, giving:

\begin{theorem}\label{Thm:MinCutOracle}
  With $O(n\log^4n)$ time and $O(n\log n)$ space for preprocessing,
  the weight of a min $st$-cut between for any two given vertices $s$
  and $t$ of an $n$-vertex planar, undirected graph with non-negative
  edge weights can be reported in constant time.
\end{theorem}

In Section~\ref{sec:MinCutReport}, we will show how to explicitly find
the cycles given the complete region tree, giving the following results:

\begin{theorem} \label{thm:mincuts} Without an increase in
  preprocessing time or space, the min $st$-cut oracle of
  Theorem~\ref{Thm:MinCutOracle} can be extended to report cuts in
  time proportional to their size.
\end{theorem}

\begin{theorem}\label{Thm:CycleBasisExplicit}
  The minimum cycle basis of an undirected planar graph with non-negative
  edge weights can be computed in $O(n\log^4n + C)$ time and $O(n\log n + C)$ space,
  where $C$ is the total size of the cycles in the basis.
\end{theorem}

\section{Separating a pair of faces}\label{sec:SepFacePair}

In this section we prove Theorem~\ref{thm:RegionSubpiece}, we show how
to find the minimum $fg$-separating cycle for the unique pair of
unseparated faces $f,g$ in a region subpiece $P_R$.  We emulate the
minimum-separating cycle algorithm due to Reif~\cite{Reif83}
(Section~\ref{sec:reifintro}).  Recall that Reif's algorithm finds the
minimum $fg$-separating cycle by first cutting the graph $G$ open
along the shortest $f$-to-$g$ path $X$, creating $G_X$, and then
computing shortest paths $C_x$ between a vertex $x \in X$ and its copy
in $G_X$.  $C_x$ is a cycle that crosses $X$ exactly once; the min
$fg$-separating cycle is the minimum over all such cycles.
 
We cannot afford to work in $G$ or $G_X$ and wish to find
shortest paths by using the precomputed distances in the dense
distance graphs using the adaptation of Dijkstra's algorithm to these
dense distance graphs developed by Fakcharonphol and Rao, FR-Dijkstra.
In order to make use of FR-Dijkstra and the precomputed distances, we
must deal with the following peculiarities:
\begin{itemize}
\item $\extDDG$ corresponds to distances in $G$, not $G_X$.  We compute
  modified dense distance graphs to account for this in
  Section~\ref{sec:modify-dense-dist}.
\item For a vertex $x \in X \cap P_R$, $C_x$ may not be contained by
  $P_R$.  We call such cycles \emph{internal cycles}.  We show how to
  find these cycles in Section~\ref{sec:find-intern-cycl}.
\item $X$ is not contained entirely in $P_R$.  For a vertex $x \in
  X\setminus P_R$, we will compute $C_x$ by composing distances in
  $\extDDG$ and $\intDDG$ between restricted pairs of boundary
  vertices of $P_R$.  We call such cycles \emph{external cycles}.  We
  show how to find these cycles in Section~\ref{sec:find-extern-cycl}.
\end{itemize}
Internal and external cycles are illustrated in
Figure~\ref{fig:min-separating-cycle}.  The minimum length cycle over
all internal and external cycles is the minimum $fg$-separating cycle
in $G$.

\begin{figure}[ht]
  \centering
  \includegraphics[scale=0.4]{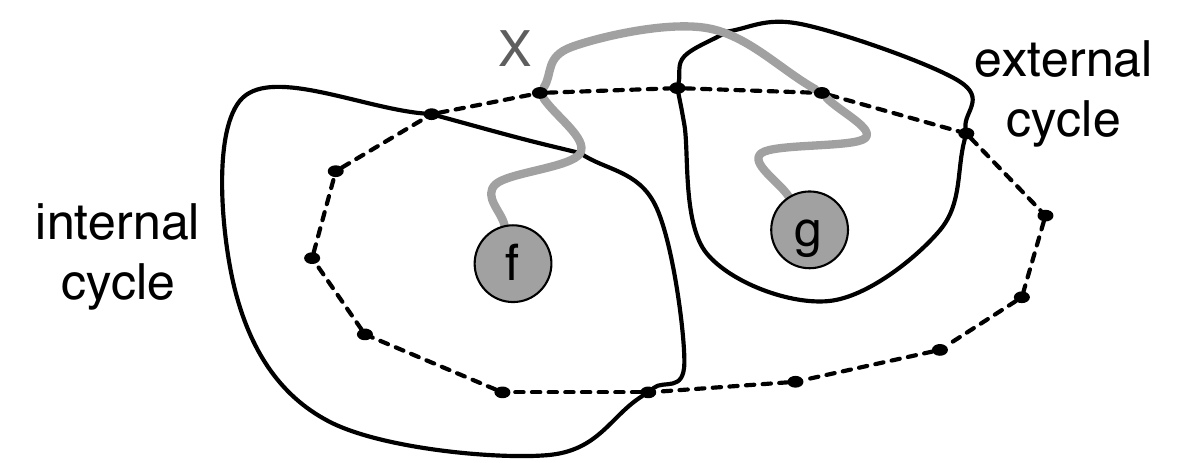}
  \caption{An external and internal cycle separating faces $f$ and $g$
    in a region subpiece, whose boundary is dashed.}
  \label{fig:min-separating-cycle}
\end{figure}

\subsection{Modifying the external dense distance graph}
\label{sec:modify-dense-dist}

We use $\extDDG$ to compute $\extDDG_X$, the dense distance graph that
corresponds to distances between boundary vertices of $P_R$ when the
graph is cut open along $X$. However, we do not compute $\extDDG_X$
explicitly; rather, we determine its values as needed.  Recall the
$\extDDG$ is really a set of dense distances graphs, one for each hole
of $P$ and one for the outer boundary of $P$.  We do the following for
each dense distance graph in the set.

Let $B$ be the set of boundary vertices of $P_R$ corresponding to a
hole or outer boundary of $P$.  Cutting $G$ open along $X$ duplicates
vertices of $B$ that are in $X$, creating $B'$.  $\extDDG_X$ can be
represented as a table of distances between every pair of vertices of
$B'$:
\[
\extDDG_X(x,y) =
\left\{
  \begin{array}{ll}
    \infty & \mbox{if $x$ is a copy of $y$} \\
    \infty & \mbox{if $x$ and $y$ are separated in $G_X$ outside
      $P_R$} \\
    \extDDG(x,y) & \mbox{otherwise}
\end{array}
\right.
\]

We describe how to determine if $x$ and $y$ are separated in $G_X$
outside $P_R$.  The portions of $X$ that appear outside $P_R$ form a
parenthesis of (a subset of) the boundary vertices, illustrated in
Figure~\ref{fig:parenthesis}.  By walking along $X$ we can label the
start and endpoints of these parentheses.  By walking along the
boundary of the subpiece we can label a group of boundary vertices
that are not separated by $X$ by pushing the vertices onto a stack
with a label corresponding to the start of a parenthesis and popping
them off when the end of the parenthesis is reached, labelling the
boundary vertices with the corresponding parenthesis.  Two boundary
vertices are not separated if they have the same parenthesis label.
Hence, whenever we are asked for a distance in $\extDDG_X(x,y)$ we
return $\infty$ if $x$ and $y$ are not in the same parenthesis and
return $\extDDG(x,y)$ otherwise.  Computing the parentheses for all
the external dense distance graphs (corresponding to the holes and
outer boundary of $P$) takes $O(|\partial P_R|)$ time.

\begin{figure}[ht]
  \centering
  \includegraphics[scale=.3]{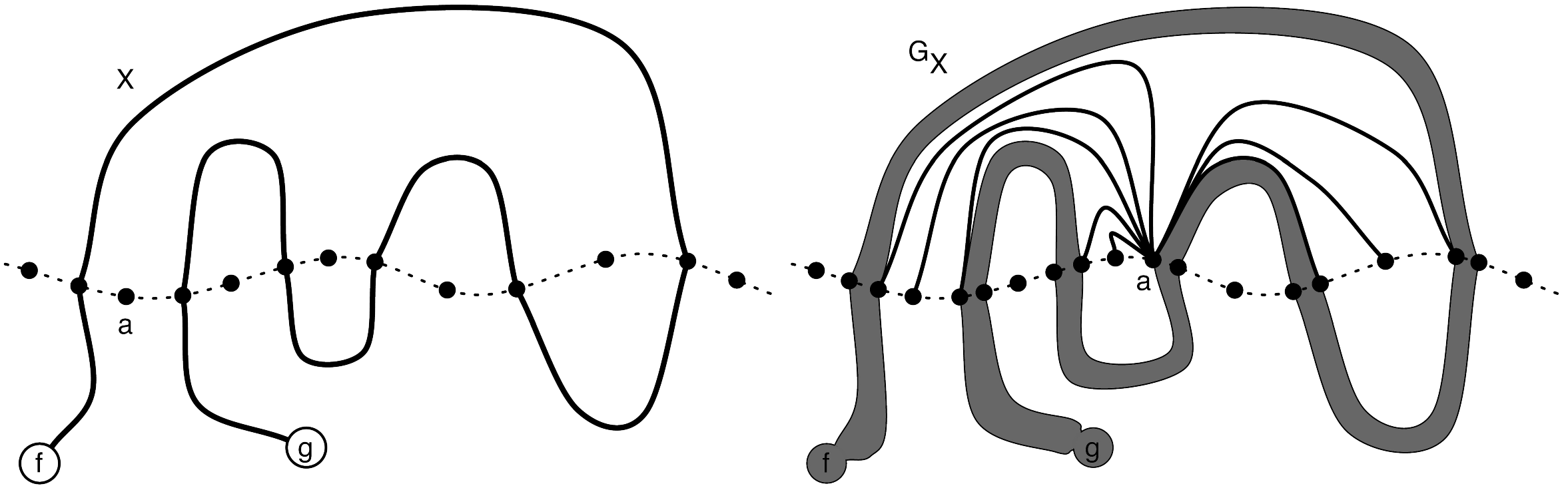}
  \caption{Modifying the external dense distance graph.  (Left) $X$ is
    given by the solid line and the boundary of the subpiece is given
    by the dotted line.  The parts of $X$ outside the subpiece form a
    parenthesis.  (Right) In $G_X$, the only finite distances from $a$
    in $\extDDG_X$ correspond to the thick lines.  The shaded area
    represents the new face created by cutting along $X$.}
  \label{fig:parenthesis}
\end{figure}

\subsection{Finding internal cycles}
\label{sec:find-intern-cycl}

Consider $D = X \cap P_R$ according to the order of the vertices
along $X$.  For each vertex $x \in D$, we compute the shortest
$x$-to-$x'$ paths in $G_X$ where $x'$ is the copy of $x$ in $G_X$.  We
do this using (standard) Dijsktra's algorithm on the cut-open graph
induced by the vertices in $P_R$ (i.e.~$G_X[P_R]$) and the modified
dense distance graph: $\extDDG_X$.  Each cycle can then be found in
$O((|P_R| + |\partial P_R|^2)\log |P_R|)$ time.  Let $x_m$ be the
midpoint vertex of $D$ according to the order inherited from $X$.
$C_{x_m}$ splits $P_R$ and $\extDDG_X$ into two parts (not necessarily
balanced). Recursively finding the cycles through the midpoint in each
graph part results in $\log |D|$ levels for a total of $O((|P_R| +
|\partial P_R|^2)\log |P_R|\log|D|) = O((|P_R| + |\partial
P_R|^2)\log^2 |P_R|)$ time to find all the internal cycles.  In order
to properly bound the running time, one must avoid reproducing long
paths in the subproblems: in a subproblem resulting from divide and
conquer, we remove degree-two vertices by merging the adjacent edges
(in the same way as Reif does~\cite{Reif83}).

\subsection{Finding external cycles}
\label{sec:find-extern-cycl}

In the lemma below, we show that every external cycle is composed of a
single edge $ab$ in the {\em unmodified} $\extDDG$ and a shortest path
$\pi_{ab}$ between boundary vertices of $P_R$ in $G$ that does not
cross $X$.  Given $\extDDG_X$ and $\intDDG_X$, a shortest path tree in
$\extDDG_X\cup\intDDG_X$ rooted at a vertex of $\partial P_R$ can be
found in $O(|\partial P_R|\log^2|P_R|)$ time using the FR-Dijkstra
algorithm.  Therefore all shortest paths, $\pi_{ab}$, can be found in
$O(|\partial P_R|^2\log^2|P_R|)$ time.

We can find $\intDDG_X$ in $O((|\partial P_R|^2+|P_R|)\log^3 |P_R|)$
time by cutting open $X$ and using the recursive internal dense
distance graph algorithm of Fakcharoenphol and Rao.  We compute
$\intDDG_X$ from scratch because $X$ has been cut open and because
$P_R$ is no longer a subgraph of $G$ due to the compact representation
we present in Section~\ref{sec:SharedEdges}.

In order to compute all the external cycles, one enumerates over all
pairs $a,b$ of vertices, 
summing the weight of $\pi_{ab}$ and the weight of edge $(a,b)$ in
$\extDDG$.  Since there are $O(|\partial P_R|)$ boundary vertices,
there are $O(|\partial P_R|^2)$ pairs to consider. The minimum-weight
external cycle then corresponds to the pair with minimum weight. By
the above, this cycle can be found in $O((|\partial
P_R|^2+|P_R|)\log^2 |P_R|)$ time.

It remains to prove the required structure of the external cycles.

\begin{lemma}
\label{lemma-single-edge}
  The shortest external cycle is composed of a single edge $ab$ in the
  {\em unmodified} $\extDDG$ and a shortest path $\pi_{ab}$ between
  boundary vertices of $P_R$ in $G$ that does not cross $X$.
\end{lemma}

\begin{proof}
  Let $C$ be the shortest external cycle that separates faces $f$ and
  $g$ -- as illustrated in Figure~\ref{fig:external}. By Theorem~\ref{thm:cross-cycle}, $C$ is a cycle that crosses
  $X$ exactly once, say at vertex $x$.  Further, $C$ is a shortest
  path $P$ between duplicates of $x$ in the graph $G_X$.  Since $C$ must
  separate $f$ and $g$, $C$ must enter $P_R$.  Starting at $x$ and
  walking along $C$ in either direction from $X$, let $a$ and $b$ be
  the first boundary vertices that $C$ reaches.  Let $\pi_{ab}$ be the
  $a$-to-$b$ subpath of $C$ that does not cross $X$.  Since $C$ is a
  shortest path in $G_X$, $\pi_{ab}$ is the shortest path between
  boundary vertices as given in the theorem.

  Let $\pi_x$ be the $a$-to-$b$ subpath of $C$ that does cross $X$.
  By definition of $a$ and $b$, $\pi_x$ contains no vertices of $P_R$
  except $a$ and $b$.
  Further, $\pi_x$ must be the shortest such path, as otherwise, $C$ would not be the shortest
  $fg$-separating cycle. Note that every path from $a$ to $b$ that does not
  contain other vertices from $P_R$ has to cross $X$ at least once.
  Therefore, $\pi_x$ must correspond to the edge $ab$ in $\extDDG$.
\end{proof}

\begin{figure}[ht]
   \centering
   \includegraphics[scale=0.8]{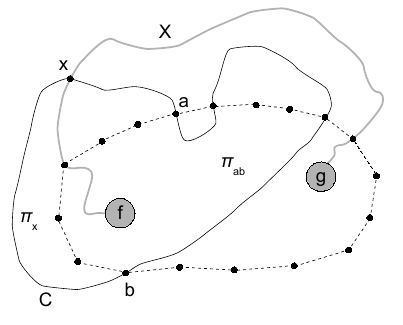}
   \caption{The shortest external cycle $C$ that separates faces $f$ and $g$.
   The figure illustrates the proof of Lemma~\ref{lemma-single-edge}.}
   \label{fig:external}
 \end{figure}

\section{Finding region subpieces}\label{sec:RegionSubpieces}

In Section~\ref{sec:overview}, we defined the region subpieces of a
piece as the intersection between a region and a piece
(Figure~\ref{fig:region-subpiece}).  In this section, we show how to
identify the region subpieces and the edges that are in them so that
we may use the min-separating cycle algorithm presented in the
previous section.  We
start by identifying the set of regions $\mathcal R_P$ whose
corresponding region subpieces of piece $P$ each contain a pair of
unseparated faces (Section~\ref{sec:ident-region}).  For each region
$R \in \mathcal R_P$ we initialize the corresponding region subpiece
$P_R$ as an empty graph.  For each edge $e$ of $P$ we determine to
what region subpieces $e$ belongs using lowest common-ancestor and
ancestor-descendent queries in the region tree
(Section~\ref{sec:ident-edges}).
We show how to do all this in $O(|P|
\log^2 n)$ time (Section~\ref{sec:SharedEdges}), proving Theorem~\ref{thm:RegionSubpiece}.

\subsection{Identifying region subpieces}\label{sec:ident-region}

Since each edge is on the boundary of two faces, we start by marking
all the faces of $G$ that share edges with $P$ in $O(|P|)$ time.  Since
a pair of unseparated faces in $P$ are siblings in the region tree, we
can easily determine the regions that contain unseparated faces. So,
in $O(|P|)$ time we can identify $\mathcal R_P$, the set of regions with
unseparated faces in $P$.  The intersection of a bounding cycle of a
region in $\mathcal R_P$ and $P$ are subpaths between pairs of
boundary vertices of $P$.  We call these paths {\em cycle paths}.  We
will need the following bound on the size of $\mathcal R_P$ in our
analysis, illustrated in Figure~\ref{fig:bound}.

\begin{figure}[ht]
  \centering
  \includegraphics[scale=0.5]{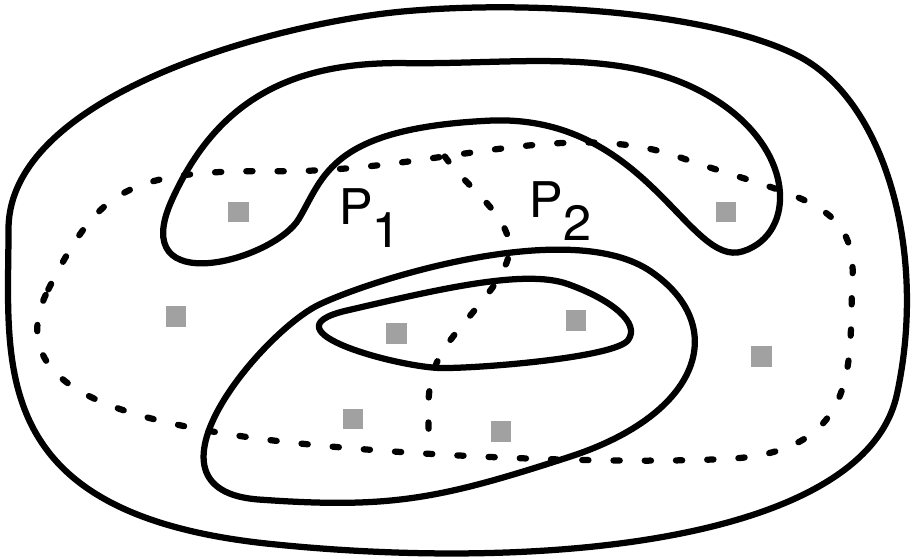}
  \caption{A piece $P$ is given by the boundaries (dotted) by its two
    child pieces $P_1$ and $P_2$.  $\mathcal R_P$ is a nesting set
    (solid cycles), each containing a pair of unseparated faces
    (grey).  Since these faces must be separated by the child pieces,
    each bounding cycle (except for one outer cycle) in $\mathcal R_P$
    must cross the dotted lines, resulting in a bound on $|\mathcal
    R_P|$.}
  \label{fig:bound}
\end{figure}

\begin{lemma}\label{Lem:Root-r-subproblems}
  $|\mathcal R_P| = O(|\partial P_1\cup\partial P_2|)$.
\end{lemma}
\begin{proof}
  Let $B$ be the set of boundary vertices in $P_1$ and $P_2$. Let $\mathcal F_1$ be
  the set of faces containing edges of $P_1$ and not
  edges of $P_2$ and let $\mathcal F_2$ be the set of faces
  containing edges of $P_2$ and not edges of $P_1$.

  Any region $R\in\mathcal R_P$ must contain at least
  one face of each $\mathcal F_1$ and $\mathcal F_2$. So, if $C$ is
  the cycle bounding $R$, either $\intc{C}$ contains $B$ or $C$
  {\it crosses} $B$, i.e., crosses the bounding faces of $P_1$ or $P_2$.

  If $\intc{C}$ contains $B$ then we claim that no
  other cycle bounding a region in $\mathcal R_P$ has this property.
  To see this, suppose for the sake of contradiction that there is
  another such cycle $C'$ bounding a region $R'\in\mathcal R_P$.
  The cycles have to nest, i.e., either $\intc{C}\subset\intc{C'}$ or
  $\intc{C'}\subset\intc{C}$. Assume w.l.o.g.\ the former. Then all
  faces of $\mathcal F_1$ are contained in the interior of a face of
  $R'$. But this contradicts the assumption that $R'$ contains at
  least one face from $\mathcal F_1$.

  We may therefore restrict our attention to regions $R \subseteq
  \mathcal R_P$ whose bounding cycle $C$ crosses $B$.

  Consider two vertices $u$ and $v$ in $B$ we and let $\mathcal R_P^{u,v}$ be the
  set of regions with cycle-paths from $u$ to $v$. By Lemma~\ref{Lem:Isometric}
  we know that if two isometric cycles share two vertices, then they share a path
  between these vertices as well. Hence, all cycles bounding regions in  $\mathcal R_P^{u,v}$
  share a path $\pi$ from $u$ to $v$. As a result at most two of the regions
  in $\mathcal R_P^{u,v}$ can contain a face adjacent to $\pi$.

  By Lemma~\ref{Lem:Isometric} no bounding cycles of regions in $\mathcal R_P$ cross, so
  each cycle going through non-consecutive boundary vertices splits the set of boundary
  vertices into two parts -- inside and outside.  By a ``chocolate-breaking''
  argument there cannot be more than $|B|-1$ such pairs of non-consecutive vertices used by cycles.
  Moreover, there are no more than $|B|$ consecutive pairs possible. As argued
  above each of these pairs cannot be used by more than two regions that contain a vertex
  inside $P$, so there are no more than $4|B|+1 = O(|\partial P_1\cup\partial P_2|)$ regions in $\mathcal R_P$.
\end{proof}

\subsection{Identifying edges of region subpieces}\label{sec:ident-edges}

Region subpieces are composed of two types of edges: {\em internal
  edges} and {\em boundary edges}.  Let $R$ be a region and let $C$ be
the bounding cycle of $R$.  An edge $e$ is an internal edge of a
region subpiece $R$ if the faces on either side of $e$ are enclosed by
$C$.  An edge $e$ is a boundary edge of $R$ if $e$ is an edge of $C$.
Every edge is an internal edge for exactly one region subpiece and we
can identify this region (Lemma~\ref{Lem:EdgeLCA}).  We can
also determine if an edge is a boundary edge for some region
(Lemma~\ref{Lem:ExtFaceLCA}).  While an edge can be a boundary edge
for several region subpieces, we can bound the potential blow-up in
running time due to this (Section~\ref{sec:SharedEdges}).

\begin{lemma}\label{Lem:EdgeLCA}
  Let $e$ be an edge of $G$ and let $f_1$ and $f_2$ be the faces
  adjacent to $e$.  Then $e$ is an internal edge of a region $R$ iff
  $R$ is the lowest common ancestor of $f_1$ and $f_2$ in the region
  tree.
\end{lemma}

\begin{proof}
  There must exist some region $R$ for which $e$ is an internal edge. Let $C$ be
  its bounding cycle.  Then both $f_1$ and $f_2$ are contained in
  $\intc{C}$ and it
  follows that $R$ must be a common ancestor of $f_1$ and $f_2$. If
  $R'$ is another common ancestor and $R$ is an ancestor of $R'$,
  then $R'$ is contained in a face of $R$, so $e$ cannot belong to
  $R$, a contradicting the choice of $R$.
\end{proof}

Iterating over each edge $e$ of $P$, we can identify the region $R$
for which $e$ is an internal edge and, if $R \in \mathcal R_P$, the
corresponding region subpiece $P_R$.  The total time required is $O(|P|
\log n)$.

\begin{lemma}\label{Lem:ExtFaceLCA}
  Let $e$ be an edge of $G$ and let $f_1$ and $f_2$ be the faces
  adjacent to $e$.  Let $R'$ be the lowest common ancestor of $f_1$ and
  $f_2$ in the region tree.  Then $e$ is a boundary edge of a region $R$
  iff $R$ is a descendant of $R'$ and exactly one of $f_1,f_2$ is a descendant
  of $R$.
\end{lemma}
\begin{proof}
  Assume first that $e\in C$, where $C$ is the cycle bounding
  $R$. Then w.l.o.g.\ $f_1\in\intc{C}$ and $f_2\in\extc{C}$. Then
  $f_1$ is a descendant of $R$ and, since $f_2$ is not, $R$ must be a
  descendant of $R'$.

  Now assume that $R$ is a descendant of $R'$ and that, say, $f_1$ is a descendant of $R$. Then $f_2$ is not a
  descendant of $R$ since otherwise, $R'$ could not be an ancestor of $R$. This implies that $e\in C$.
\end{proof}

Let $R$ be a region in $\mathcal R_P$ and let $C$ be the bounding
cycle of $R$.  Let $P_1$ and $P_2$ be the children of $P$ and let $B$
be the union of boundary vertices of $P_1$ and $P_2$. 
Consider the following algorithm to find starting points of cycle paths.
\begin{description}
\item[Cycle path starting points identification algorithm]
  Pick a boundary vertex $u\in B$.  For every edge $e$ adjacent to $u$
  (there are at most three such edges), check to see if $e$ is a
  boundary edge of $R$.  If there is
  no such edge, then there is no cycle path through $u$.
  Otherwise, mark $e$ as a starting point of a cycle path for $R$.
  Repeat this process for every vertex in $B$.
\end{description}
Using Lemma~\ref{Lem:ExtFaceLCA} and the top tree, this process
takes $O(|B|\log n)$ time for each region $R$ since a constant number of
tree queries for every vertex in $B$. By
Lemma~\ref{Lem:Root-r-subproblems}, repeating this
for all regions in $\mathcal R_P$ takes $O(|B|^2\log n)$ time.

After identifying starting points for cycle-path we can find all edges belonging to
them using {\it linear search}, i.e., the next edge on the cycle $C$ is found by
looking at the endpoint of the previous edge and checking which of the two remaining
edges belongs to $C$.  If the cycles are edge-disjoint over all regions $R\in\mathcal R_P$,
then the cycle paths will also be edge-disjoint. In such a case the time to find all the
region subpieces using linear search is $O((|B|^2 + |P|)\log n)$.
However, the cycles are not necessarily edge disjoint.  We overcome this complication
in the next section.

\subsection{Efficiently identifying boundaries of region
  subpieces}\label{sec:SharedEdges}

Since cycles will share edges, the total length of cycle paths over
all cycles can be as large as $O(|P|^{3/2})$.  However, we can maintain
the efficiency of the cycle path identification algorithm by using a
compact representation of each cycle path.  The compact representation
consists of edges of $P$ and {\em cycle edges} that represent paths in
$P$ shared by multiple cycle paths.

View each edge of $G$ as two oppositely directed darts and view the
cycle bounding a region as a clockwise cycle of darts.  The following is
a corollary of Lemma~\ref{Lem:Isometric}.

\begin{corollary}\label{Lem:SharedEdges}
  If two isometric cycles $C$ and $C'$ of $G$ share a dart, then
  either $\intc{C}\subseteq\intc{C'}$ or $\intc{C'}\subseteq\intc{C}$.
\end{corollary}

Let $\mathcal F$ be the forest representing the ancestor/descendant
relationship between the bounding cycles of regions in $\mathcal
R_P$. By Lemma~\ref{Lem:Root-r-subproblems}, there are $O(\sqrt r)$
bounding cycles and, since we can make descendent queries in the region
tree in $O(\log
n)$ time per query, we can find $\mathcal F$ in $O(r\log n)$ time. Let $d$ be
the maximum depth of a node in $\mathcal F$ (roots have depth
$0$). For $i = 0,\ldots,d$, let $\mathcal C_i$ be the set of cycles
corresponding to nodes at depth $i$ in $\mathcal F$. By
Corollary~\ref{Lem:SharedEdges} and
Lemma~\ref{Lem:SharedEdges}:
\begin{corollary}\label{Lem:DisjointCycles}
  For any $i\in\{0,\ldots,d\}$, every pair of cycles in $\mathcal C_i$
  are pairwise dart-disjoint.
\end{corollary}

\subsubsection*{Bottom-up algorithm}

We find cycle paths for cycles in $\mathcal C_d$, then $\mathcal
C_{d-1}$, and so on.  The cycles in $\mathcal C_d$ are dart disjoint,
so any edge appears in at most two cycles of $\mathcal C_d$.  We find
the corresponding cycle paths using the cycle path identification
algorithm in near-linear time.  While
Corollory~\ref{Lem:DisjointCycles} ensures that the cycles in
$\mathcal C_d$ are mutually dart-disjoint, they can share darts
with cycles in $C_{d-1}$.  In order to efficiently walk along subpaths
of cycle paths $Q$ that we have already discovered, we use a balanced binary
search tree (BBST) to represent $Q$. We augment the BBST to store in each node
the length of the subpath it represents. Now, given two nodes in $Q$, the
length of the corresponding subpath of $Q$ can be determined in logarithmic time.

\begin{figure}[ht]
  \centering
  \includegraphics[scale=0.4]{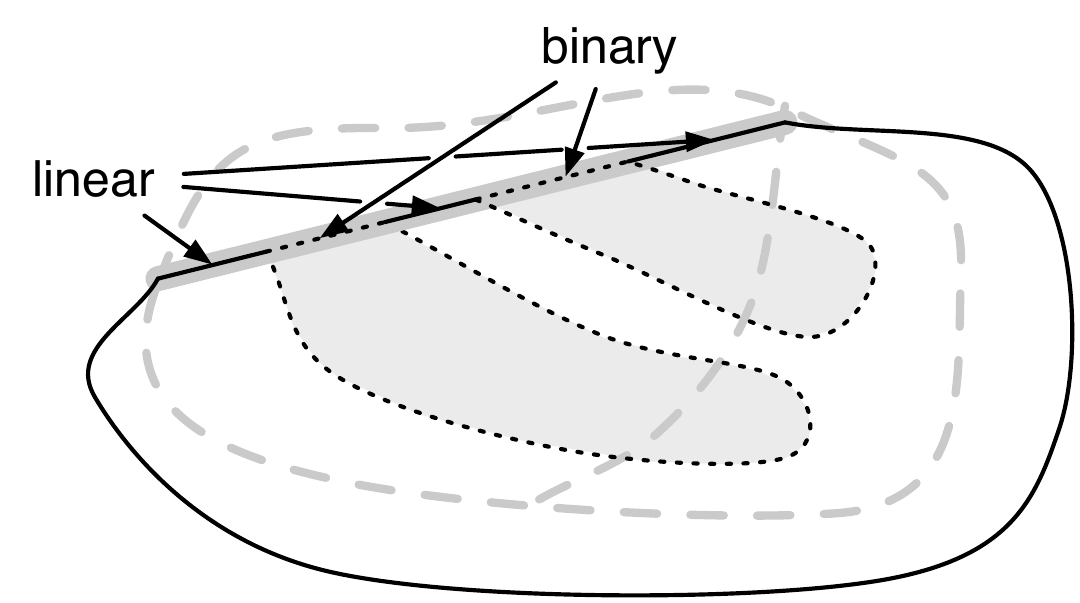}
  \caption{Finding a cycle path (highlighted straight line) for a
    cycle $C \in \mathcal C_{d-1}$ between boundary vertices of $P_1$ and
    $P_2$ (grey dashed lines) is found by alternating linear (solid)
    and binary (dotted) searches.  Binary searches corrrespond to
    cycle paths of region subpieces (shaded) bounded by cycles in
    $\mathcal C_d$.}
  \label{fig:alt-search}
\end{figure}

To find the cycle paths of a cycle $C\in \mathcal C_{d-1}$ that bounds
a region $R$, we emulate the cycle path identification algorithm:
start walking along a cycle path $Q$ of $C$, starting from a vertex of
$B$, and stop if you reach an edge $e = uv$ that has already been
visited {\em (linear search)}.  In this case, $e$ must be an edge of a
cycle path $Q'$ of a cycle $C' \in \mathcal C_d$.  By
Lemma~\ref{Lem:Isometric}, the intersection of $Q$ and $Q'$ is a
single subpath and so we can use the BBST to find the last vertex $w$
common to $Q$ and $Q'$ {\em (binary search)}.  We add to $P_R$ an edge
$uw$ of length equal to the length of the $u$-to-$w$ subpath of $Q$ to
compactly represent this subpath. If $w \in B$, we stop our walk along
$Q$.  Otherwise we continue walking (and adding edges to the
corresponding region subpiece) in a linear fashion, alternating
between linear and binary searches until a boundary vertex is
reached.  See Figure~\ref{fig:alt-search}.

We have shown how to obtain region subpieces for cycles in $\mathcal
C_d$ and in $\mathcal C_{d-1}$. In order to repeat the above idea to
find cycle paths for cycles in $\mathcal C_{d-2}$, we need to build
BBSTs for cycle paths of cycles in $\mathcal C_{d-1}$.  Let $Q$ be one
such cycle path.  $Q$ can be decomposed into subpaths $Q_1Q_1'\cdots
Q_kQ_k'$, where $Q_1,\ldots,Q_k$ are paths obtained with linear
searches and $Q_1',\ldots,Q_k'$ are paths obtained with binary
searches (possibly $Q_1$ and/or $Q_k'$ are empty). To obtain a binary
search tree $\mathcal T$ for $Q$, we start with $\mathcal T$ the BBST
for $Q_1$.  We extract a BBST for $Q_1'$ from the BBST we used to find
$Q_1'$ and merge it into $\mathcal T$. We continue merging with BBSTs
representing the remaining subpaths.

Once BBSTs have been obtained for cycle paths arising from
$\mathcal C_{d-1}$, we repeat the process for cycles in $\mathcal
C_{d-2},\ldots,\mathcal C_0$.

\subsubsection*{Running time}
We now show that the bottom-up algorithm runs in $O((|B|^2 + |P|)\log^2n)$ time
over all region subpieces, proving Theorem~\ref{thm:RegionSubpiece}. We have already described how to identify
boundary vertices that are starting points of cycle paths within this time bound.
It only remains to bound the time required for linear and
binary searches and BBST construction.

A subpath identified by a linear search consists only of edges that
have not yet been discovered.  Since each step of a linear search
takes $O(\log n)$ time, the total time for linear searches is $O(|P|\log
n)$.

The number of cycle paths corresponding to a cycle $C$ is bounded by the
number of boundary vertices, $O(|B|)$.  We consider three types of
cycle paths.  Those where
\begin{enumerate}
\item all edges are shared by a single child of $C$ in $\mathcal F$,
\item no edges are shared by a child, and
\item some but not all edges are shared by a single child.
\end{enumerate}
Cycle paths of the first type are identified in a single binary search
which, by Lemma~\ref{Lem:Root-r-subproblems}, sums up to a
total of $O(|B|^2)$ binary searches over all cycles $C\in\mathcal
F$.  Cycle paths of the second type do not require binary search.  For
a cycle path $Q$ in the third group, $Q$ can only share one subpath
with each child (in $\mathcal F$) cycle by Lemma~\ref{Lem:Isometric};
hence, there can be at most two binary searches per child.  Summing over all such
cycles, the total number of binary searches is $O(|B|)$ by
Lemma~\ref{Lem:Root-r-subproblems}.

In total there are $O(|B|^2)$ binary searches.  Each BBST has $O(|P|)$ nodes.
In traversing the binary search tree, an edge is checked for
membership in a given cycle path using Lemma~\ref{Lem:ExtFaceLCA} in
$O(\log n)$ time.  Each binary search therefore takes $O(\log |P| \log
n) = O(\log^2 n)$ time so the total time spent performing binary
searches is $O(|B|^2\log^2n)$.

It remains to bound the time needed to construct all BBSTs.  We merge BBSTs $T_1$
and $T_2$ in $O(\min\{|T_1|,|T_2|\}\log(|T_1| + |T_2|\})) =
O(\min\{|T_1|,|T_2|\}\log n)$ time by inserting elements from the
smaller tree into the larger.

When forming a BBST for a cycle path of a cycle $C$, it may be
necessary to delete parts of cycle paths of children of $C$. By
Lemma~\ref{Lem:Isometric}, these parts intersect $\into{C}$ and will
not be needed for the remainder of the algorithm. The total number of
deletions is $O(|P|)$ and they take $O(|P|\log |P|)$ time to execute.  So,
ignoring deletions, notice that paths represented by BBSTs are pairwise
dart disjoint (due to Corollary~\ref{Lem:DisjointCycles}).  Applying
Lemma~\ref{Lem:Merge} with $k = \log n$ and $W = r$ then gives
Theorem~\ref{thm:RegionSubpiece}.

\section{Adding a separating cycle to the region
  tree}\label{sec:UpdateRegionTree}

Above, we showed how to find a compact representation of a minimum cycle $C$ separating a pair of faces in a region $R$.
This cycle should be added to the basis we are constructing and in this section, we show how to update the region tree
$\mathcal T$ accordingly. As in the previous section, let $P_R$ be the
region subpiece $P\cap R$ of piece $P$ associated with region $R$.

When $C$ is added to the partial basis, $R$ is split into two regions,
$R_1$ and $R_2$.  Equivalently, in $\mathcal T$, $R$ will be replaced
by two nodes $R_\ell$ and $R_r$. The children $\mathcal F$ of $R$ will
be partitioned into children $\mathcal F_\ell$ of $R_\ell$ and
$\mathcal F_r$ of $R_r$.  Define $R_\ell$ to be the region as defined
by the children of $R$ that are contained to the left of $C$ (and
symmetrically define $R_r$).  We describe an algorithm that finds
$\mathcal F_\ell$ and detects whether $\mathcal F_\ell$ is contained
by $\intc{C}$ or $\extc{C}$.  Finding $\mathcal F_r$ is symmetric.
The algorithms take $O(|\mathcal F_\ell|\log^3n + (|P_R| + |\partial
P_R|^2)\log n)$ and $O(|\mathcal F_r|\log^3n + (|P_R| + |\partial
P_R|^2)\log n)$ time and so we can identify the smaller side of the
partition in $O(\min\{|\mathcal F_\ell|,|\mathcal F_r|\}\log^3n+
(|P_R| + |\partial P_R|^2)\log n)$ time, as required for
Theorem~\ref{thm:updateregiontree}.

Given the smaller side of the partition, we use cut-and-link
operations to update $\mathcal T$ in $O(\min\{|\mathcal
F_\ell|,|\mathcal F_r|\}\log n)$ additional time, thus proving
Theorem~\ref{thm:updateregiontree}.  See
Figure~\ref{fig:region-tree-update} for an illustration.  Assume,
w.l.o.g., that $\mathcal F_\ell$ is the smaller set. If $\mathcal
F_\ell$ is contained by $\intc{C}$ then update $\mathcal T$ by:
cutting the edges between $R$ and each element in $\mathcal F_\ell$,
linking each element in $\mathcal F_\ell$ to $R_\ell$, making $R$ the
parent of $R_\ell$, identifying $R$ with $R_r$. If $\mathcal F_\ell$
is contained by $\extc{C}$ then update $\mathcal T$ by: cutting the
edges between $R$ and each element in $\mathcal F_\ell$, linking each
element in $\mathcal F_\ell$ to a new node $u$, making $u$ the parent
of $R$; identifying $R$ with $R_\ell$ and $u$ with $R_r$.

\subsection{Partitioning the faces}

$R$ is represented compactly: vertices in $G[R]$ of degree 2 are
removed by merging the adjacent edges creating {\em super edges}.
Each super edge is associated with the first and last edge on the
corresponding path.  In addition to partitioning the faces, we must
find the compact representation for the two new regions, $R_\ell$ and
$R_r$.

\begin{figure}[ht]
  \centering
  \includegraphics[scale=0.4]{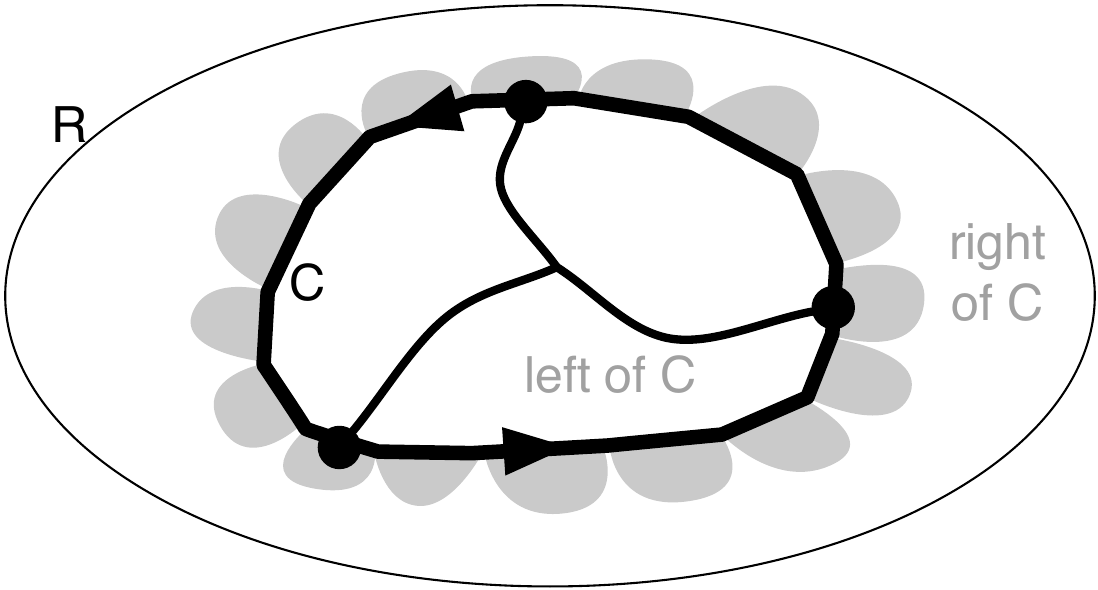}
  \caption{If $C$ (bold cycle) is a counterclockwise cycle, then
    $R_\ell$ is contained by $\intc{C}$. The children of $R$ (boundary
    given by thin cycle) adjacent and to the right of $C$ are grey.
    In this example,
    the edges to the left of $C$ (and not on $C$) will never reach a
    boundary edge of $R$: therefore the left of $C$ forms $\intc{C}$.
    Vertices of $L$ are given by dark circles.}
  \label{fig:partition-faces}
\end{figure}

The algorithm for finding $\mathcal F_\ell$ starts with an empty set and
consists of three steps:
\begin{description}
\item[Left root vertices] Identify the set $L$ of vertices $v$ on $C$
  having an edge emanating to the left of $C$; also identify, for each
  $v\in L$, the two edges on $C$ incident to $v$ (in $G$, not the
  compact representation).  (Details will be given in
  Section~\ref{subsec:Step1}.)
\item[Search] Start a search in $R$ from each vertex of $L$.  During
  this search, avoid
  edges of $C$ and edges that emanate to the right of $C$.
  For each super edge
  $\hat e$ of $R$ visited, find the first (or last) edge $e$ on the
  path represented by $\hat e$. 
\item[Add] For each searched edge and for each pair of faces $f_1,f_2$
  adjacent to this edge, find the
  two children of $R$ in $\mathcal T$ having $f_1$ and $f_2$ as
  descendants, respectively. Add those nodes that are also descendants
  of $\mathcal F$ to $\mathcal F_\ell$. 
\end{description}

This algorithm correctly builds $\mathcal F_1$: The algorithm visits
all super edges $\hat e$ that are strictly inside $R$ and on the left
side of $C$.  Let $A_1$ and $A_2$ be the children of $R$ that are
added corresponding to edge $e$.  $A_i$ is a region or a face of $G$: let
$C_i$ be the bounding cycle.  Since $f_i$ is a descendent of $A_i$,
$f_i$ is contained by $\intc{C_i}$.  Since $e$ is in $C_1$ and $C_2$:
so must $\hat e$. $A_1$ and $A_2$ are therefore the child regions of $R$ on either
side of $\hat e$.

The algorithm can easily determine if $\mathcal F_1$ is contained by
$\intc{C}$ or $\extc{C}$, by noting whether a searched edge ever
belongs to the cycle bounding $R$.  Given a searched edge $e$ and
adjacent edges $f_1$, we can determine whether $e$ is in the bounding
cycle of $R$ (Lemma~\ref{Lem:ExtFaceLCA}).  The search can only the
boundary of $R$ if $\mathcal F_\ell$ is contained by $\extc{C}$.  (See Figure~\ref{fig:partition-faces}.)

\subsubsection*{Analysis}
The above-described algorithm can be implemented in $O(|\mathcal
F_\ell|\log^3n + (|P_R| + |\partial P_R|^2)\log n)$ time. Finding the left-root
vertices is the trickiest part; while $|L| = O(|\mathcal F_\ell|)$,
$|L|$ could be much smaller than the
number of vertices in $C$, even in the compact representation. We give
details in Section~\ref{subsec:Step1}.  Assuming that left-root vertices
can be found quickly, we analyze the remaining steps.

\begin{lemma}
The number of super edges in $G[R]$ that are searched inside $C$ is $O(|\mathcal F_\ell|)$.
\end{lemma}
\begin{proof}
  Since $G$ is degree three and all faces and isometric cycles in $G$
  are simple, the compact representation of $R$ is also degree three.
  Since there are no degree-two vertices, $G[R]$ is 3-regular.
  Therefore $|E(G[R])| = \frac{3}{2}|V(G[R])|$.  Since, by Euler's
  formula, $|V(G[R])|-|E(G[R])|+|F_\ell| = 1$ we get, $|E(G[R])|= 3(|F_\ell|-1)$.
\end{proof}

\paragraph{Search step}
As we can identify if an edge belongs to $C$
(Lemma~\ref{Lem:ExtFaceLCA}) in $O(\log n)$ time, we can avoid edges
of $C$ during the search.  Since $G[R]$ is degree three, we will never
encounter edges emanating from the right of $C$.  The search can be
done by DFS or BFS in linear time, starting with vertices of $L$.
Given a super edge $\hat e$ found by this search, we find the first or
last edge $e$ (of $G$) on the path the super edge represents in $O(1)$
time, since $e$ is associated with $\hat e$.  The search takes
$O(|\mathcal F_\ell|\log n)$.

\paragraph{Add step}
Checking if $\mathcal F_\ell$ is contained by $\intc{C}$ or
$\extc{C}$) takes $O(\log n)$ time per searched edge
(Lemma~\ref{Lem:ExtFaceLCA}) given the adjacent faces $f_1$ and $f_2$
(which can be identified using the original graph).  Finding the
children of $R$ that are ancestors of $f_1$ and $f_2$ also takes
$O(\log n)$ time using the operations $\jump(R,f_1,1)$ and
$\jump(R,f_2,1)$ in the top tree for $\mathcal T$.  The total time
spent adding is $O(|\mathcal F_\ell|\log n)$.

\subsection{Finding left-root vertices}\label{subsec:Step1}

We show how to find the set $L$ of left-root vertices along $C$ in
$O(|\mathcal F_\ell|\log^3n + |C|\log n)$ time where $|C|$ is the number
of super edges in the compact representation of $C$.  Recall from
Section \ref{sec:SepFacePair} that $C$ has $O(|P_R| + |\partial P_R|^2)$
super edges and they are of three different types: edges in
$\extDDG$, edges in $\intDDG(P_R)$, and edges and cycle paths in $P_R$.  We
will show how to use binary search to prune certain super edges of $C$
that do not contain vertices of $L$.  We first assume that each super
edge is on the boundary of a child region (as opposed to a child face)
of $R$ that is to the left of $C$.  We relax this
assumption in Section~\ref{sec:face-assumption}.

The following lemma is the key to using binary search along $C$:
\begin{lemma}\label{Lem:BinSearch}
  Let $P$ be the shortest $u_1$-to-$u_2$ path in $G$ that is also a
  subpath of $C$. For $i = 1,2$, let $e_i$ be the edge of $P$ that is incident
  to $u_i$ and let $r_i$ be the child-region of $R$ that is left of
  $C$ and  bounded by $e_i$. Then $r_1 = r_2$ if and only if no
  interior vertex of $P$ belongs to $L$.
\end{lemma}
\begin{proof}
  The reverse direction is trivial.

  By our assumption, $r_1$ and $r_2$ are regions, not faces.  Their
  bounding cycles must therefore be isometric.  If $r_1 = r_2$, then
  by Lemma~\ref{Lem:Isometric}, $P$ is a subpath of the boundary of
  $r_1$: no interior vertex of $P$ could belong to $L$ in this
  case. This proves the forward direction.
\end{proof}

\subsubsection*{Shortest path covering}\label{subsubsec:SPCovering}
In order to use Lemma~\ref{Lem:BinSearch}, we cover the left-root vertices
of $C$ with two shortest paths $P$ and $Q$.  Let $r$ be a vertex that
is the endpoint of a super edge of $C$.  Since $C$ is isometric, there
is a unique edge $e$ such that $C$ is the union of $e$ and two
shortest paths $P'$ and $Q'$ between $r$ and the endpoints of $e$.
Note that $e$ could be in the interior of a super edge of $C$.  The
paths $P$ and $Q$ that we use to cover $L$ are prefixes of $P'$ and $Q'$.

To find $e$, we first find $\hat e$, the super edge that contains $e$.
Since $P$ and $Q$ are shortest paths and shortest paths are
unique, the weight of each path is at most half the weight of the
cycle.  To find $\hat e$, simply walk along the super edges of $C$ and
stopping when more than half the weight is traversed: $\hat e$ is the
last super edge on this walk.

Given $\hat e$, we continue this walk according to the type of super
edge that $\hat e$ is.  If $\hat e$ corresponds to a cycle path, then,
by definition, all the interior vertices of $\hat e$ have degree two
in $R$ and so cannot contain a left-root vertex; there is no need to
continue the walk. $P$ and $Q$ are simply the paths along $C$ from $r$
to $\hat e$'s endpoints.  This takes $O(|C|)$ time.

If $\hat e$ is an edge of $\intDDG(P_R)$ or $\extDDG$, we continue the
walk.  We describe the process for $\intDDG(P_R)$ as $\extDDG$ is
similar: we continue the walk started above through the subdivision
tree of $P_R$ that is used to find $\intDDG(P_R)$.  $\hat e$ is given
by a path of edges in the internal dense distance graph of $P_R$'s
children in the subdivision tree.  We may assume that we have a top
tree representation of the shortest path tree containing this path and so we can find the child super
edge $\hat e_c$ that contains $e$; using binary search this takes $O(\log^2
n)$ time.  Recursing through the subdivision tree finds a cycle path
or edge that contains $e$ for a total of $O(\log^3 n)$ time.

When we are done, $P$ and $Q$ are paths of super edges from  $\extDDG$
or $\intDDG(P_R)$.   $P$ and $Q$ each have $O(|C| + \log n)$ super
edges and they are found in $O(|C|+\log^3n)$ time.

\subsubsection*{Building $L$}
Using Lemma~\ref{Lem:BinSearch}, we will decompose $P$ into maximal
subpaths $P_1,\ldots,P_k$ such that no interior vertex of a
subpath belongs to $L$. Each subpath $P_i$ will be associated with the
child region of $R$ to the left of $C$ that is bounded (partly) by
$P_i$.  We repeat this process for $Q$ and find $L$ in $O(k)$ time by
testing the endpoints of the subpaths.

Let $\hat e$ be one of the $O(|C| + \log n)$ super edges of $P$.
$\hat e$ is either an edge of $\extDDG$ or $\intDDG(P_R)$.  Suppose
$\hat e$ is in $\intDDG(P_R)$.  We can apply Lemma~\ref{Lem:BinSearch}
to the first and last edges on the path that $\hat e$ represents, and
stop if there are no vertices of $L$ in the interior of the path.
Otherwise, with the top tree representation of the shortest path tree containing the
shortest path representing $\hat e$, we find the midpoint of this path
and recurse.  If $\hat e$ is in $\extDDG$, the process is similar.
Adjacent subpaths may still need to be merged after the above process,
but this can be done in time proportional to their number.

How long does it take to build $L$?  Let $\hat e$ be a super edge
representing subpath $P_{\hat e}$ of $P$ and let $m$ be the number of
interior vertices of $P_{\hat e}$ belonging to $L$.  Then there are
$m$ leaves in the recursion tree for the search applied to $\hat
e$. We claim that the height of the recursion tree is $O(\log^2n)$.
Let $S$ be some root-to-leaf path in the recursion tree.  If $\hat e$
is in $\intDDG(P_R)$, $S$ is split into $O(\log n)$ subpaths, one for
each level of the subdivision tree; in each level, the corresponding
subpath is halved $O(\log n)$ times before reaching a single edge. If
$\hat e$ is in $\extDDG$, the search may go root-wards in the
subdivision tree but once we traverse down, we are in $\intDDG$ and
will thus not go up again. The depth of the recursion tree is still
$O(\log^2n)$.

At each node in the recursion tree, we apply two top tree operations
to check the condition in Lemma~\ref{Lem:BinSearch} and one top tree
operation to find the midpoint of a path for a total of $O(\log n)$
time.  The total time spent finding the $m$ vertices of $L$ in $Q$ is
$O(m\log^3n)$ time.   If $m = 0$, we still need $O(\log n)$ time to
check the condition in Lemma~\ref{Lem:BinSearch}. Summing over
all super edges of $P$, the time required to identify $L$ is
$O(|C|\log n + |L|\log^3n) = O(|C|\log n + |\mathcal F_\ell|\log^3n)$,
as desired.

\subsubsection*{Handling faces}\label{sec:face-assumption}
We have assumed that every child of $R$ incident to and left of $C$ is
a region, not a face.  Lemma~\ref{Lem:BinSearch} is only true for this
case: boundaries of faces need not be isometric, and so the
intersection between a face and shortest path may have multiple
components.  However, notice that after the triangulation of the
primal followed by the triangulation of the dual, every face $f$ of $G$ is bounded by a simple
cycle of the form $e_1P_1e_2P_2e_3P_3$ where $e_1$, $e_2$, and $e_3$
are edges and $P_1$, $P_2$, and $P_3$ are tiny-weight shortest paths (see Section~\ref{sec:defs}).
Call the six endpoints of edges $e_1,e_2,e_3$ the \emph{corners
of $f$}. We associate each edge of $f$ with the
path containing it among the six paths $e_1$, $e_2$, $e_3$, $P_1$, $P_2$, and $P_3$.

We present a stronger version of Lemma~\ref{Lem:BinSearch} which
implies the correctness of the left-root vertex finding algorithm even
when children of $R$ are faces, not regions:
\begin{lemma}\label{Lem:BinSearchElementary}
  Let $P$ be the shortest $u_1$-to-$u_2$ path in $G$ that is also a
  subpath of $C$. For $i = 1,2$, let $e_i$ be the edge on $P$ incident
  to $u_i$ and let $r_i$ be the child of $R$ to the left of $C$ and
  containing $e_i$.
\begin{enumerate}
\item If neither $r_1$ nor $r_2$ are faces, then $r_1 = r_2$ if and
  only if no interior vertex of $P$ belongs to $L$.
\item If exactly one of $r_1,r_2$ is a face, then some interior vertex
  of $P$ belongs to $L$.
\item If both $r_1$ and $r_2$ are faces and $r_1\neq r_2$, then some interior vertex
  of $P$ belongs to $L$.
\item If both $r_1$ and $r_2$ are faces, $r_1 = r_2$, and $e_1$ and $e_2$ are
  associated with different subpaths of $r_1$, then some interior vertex of $P$ is
  a corner of $r_1$ or belongs to $L$.
\item If both $r_1$ and $r_2$ are faces, $r_1 = r_2$, and $e_1$ and $e_2$ are
  associated with the same subpath of $r_1$, then no interior vertex of $P$ belongs to $L$.
\end{enumerate}
\end{lemma}
\begin{proof}
  Part~1 is Lemma~\ref{Lem:BinSearch} and parts~2 and~3 are trivial. For
  part~4, we may assume that $P$ is fully contained in the boundary of $r_1$
  since otherwise, some interior vertex of $P$ belongs to $L$. Since $e_1$ and
  $e_2$ are associated with different subpaths of $r_1$, it follows that
  some interior vertex of $P$ is a corner of $r_1$.
  For part~5, we may assume that $e_1\neq e_2$. Then $e_1$ and $e_2$ are
  on the same (tiny weight) shortest path in $r_1$ so $P$ must be
  contained in the boundary of $r_1$. It follows that no interior vertex of $P$
  belongs to $L$.
\end{proof}

Using Lemma~\ref{Lem:BinSearchElementary} instead of
Lemma~\ref{Lem:BinSearch}, our $L$-finding algorithm will also
identify corners of faces incident to $P$. Since
each face has only 6 corners but contributes at least two vertices to
$L$, this will not increase the asymptotic running time.

\subsection{Obtaining new regions}

While we have found the required partition of the children of $R$ and
updated the region tree accordingly, it remains to find compact
representations of the new regions $R_\ell$ and $R_r$.  Recall that we
only explicitly find one side of the partition, w.l.o.g., $\mathcal F_\ell$.

To find $R_\ell$, start with an initially empty graph. In the {\em search}
step, we explicitly find all the super edges of $R_\ell$ that are not on
the boundary of $C$.  Remove these edges from $R$ and add them to
$R_1$.  The remaining super edges are simply subpaths of $C$ between
consecutive vertices of $L$.  These edges can be added to $R_\ell$ {\em
  without} removing them from $R$.

The super edges left in $R$ are exactly those in $R_r$.  However,
there may be remaining degree-two vertices that should be removed by
merging adjacent super edges.  All such vertices, by construction,
must be in $L$, and so can be removed quickly.

That super edges are associated with the first and last edges on their
respective paths is easy to maintain given the above construction.
The entire time required to build the new compact representation is
$O(|\mathcal F_\ell|)$.

\section{Reporting min cuts}\label{sec:MinCutReport}

By Theorem~\ref{Thm:MinCutOracle}, we can report the weight of any
minimum $st$-cut in constant time.  We extend this to report a minimum
separating cycle for a given pair of faces in $G^*$ in time
proportional to the number of edges in the cycle. By duality of the
min cuts and min separating cycles, this will prove
Theorem~\ref{thm:mincuts}.

In this section, we do not assume that the graph is 3-regular.  The
edges added to achieve 3-regularity may increase the number of edges in a
cycle.  However, we can still compute $\mathcal T$, the region tree of
$G^*$, with the degree-3 assumption.  We will rely only on the
relationship between faces in $G^*$, which did not change in the
construction for the degree-3 assumption.  Since cycles in the min cycle
basis are boundaries of regions represented by $\mathcal T$, the region tree
also reflects the ancestor/descendant relationships between cycles in
the min cycle basis.

Recall that we view a cycle $C$ as a clockwise cycle of {\em darts}
(Section~\ref{sec:SharedEdges}).  It follows from
Lemma~\ref{Lem:ExtFaceLCA} that the set of darts
in $C$ that are not also in a cycle $C'$ that is an ancestor of $C$ forms a path (possibly equal
to $C$). Further, the set of darts in $C$ that are not also in {\em
  any} strict ancestor of $C$ also form a path, denoted $P(C)$. Using
the next lemma, we can succinctly represent any cycle using these
paths.  See Figure~\ref{fig:succinct-cycle} for an illustration.

\begin{lemma}\label{lem:cycle-decomp}
  Let $C = C_0, C_1, C_2, \ldots$ be the ancestral path to the root of
  $\mathcal T$ for $C$.  $C$ can be written as the concatenation of
  the path $P(C)$, prefixes of $P(C_1), P(C_2), \ldots, P(C_{k-1})$, a
  subpath of $P(C_k)$ and suffixes of $P(C_{k-1}), \ldots, P(C_2),
  P(C_1)$ and in that order.
\end{lemma}

\begin{figure}[ht]
  \centering
  \includegraphics[scale=0.4]{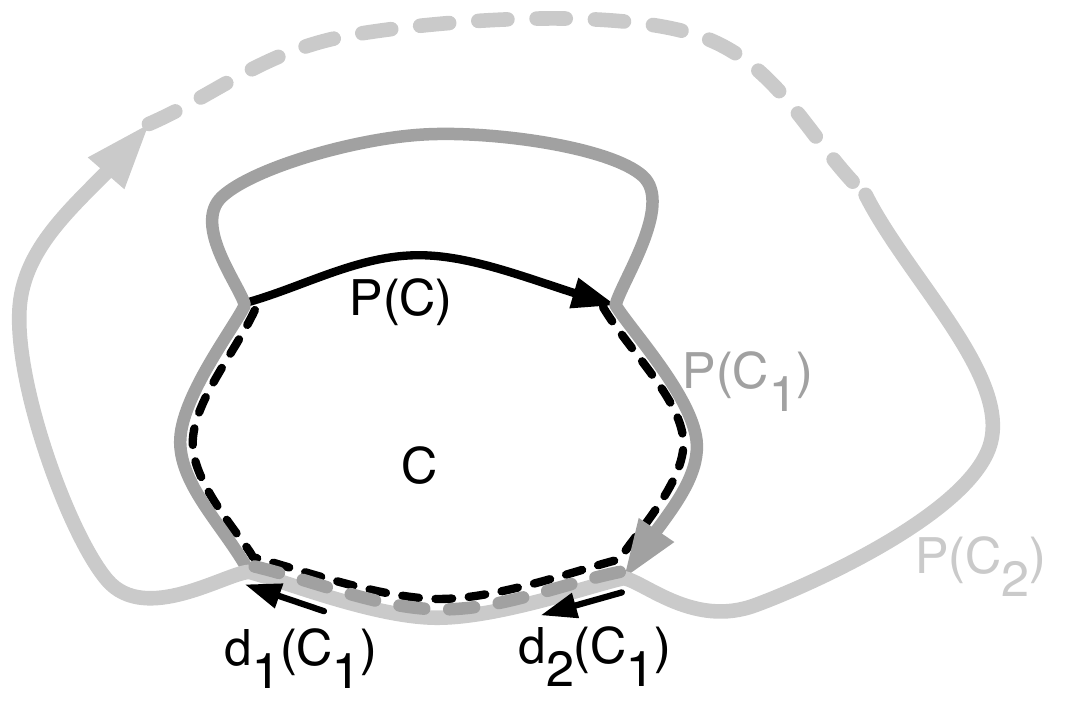}
  \caption{A succinct representation of a cycle in the min cycle
    basis.}
  \label{fig:succinct-cycle}
\end{figure}

\begin{proof}
  Let $C_k$ be the root-most cycle that shares a dart $d$ with $C$:
  $d$ is in $P(C_k)$.  By Lemma~\ref{Lem:ExtFaceLCA}, the intersection
  of $C$ with $C_k$ is a single path: it must be a subpath $P'$ of
  $P(C_k)$.  Let $Q = C \setminus P'$.  By the definition of
  $P(C_{k-1})$, the start of $Q$ must be the start of $P(C_{k-1})$ and
  the end of $Q$ must be the end of $P(C_{k-1})$.  The remainder of
  the proof follows with a simple induction.
\end{proof}

Let $d_1(C)$ be the last dart of $C$ before $P(C)$ and let $d_2(C)$ be
the first dart of $C$ after $P(C)$.  Suppose additionally that we
know, for every dart $d$, the cycle $C(d)$ for which $d \in P(C)$.

\subsection{Finding a min-separating cycle}

If we are given $d_1(C)$, $d_2(C)$, and $P(C)$ for every cycle (node)
in $\mathcal T$ and $C(d)$ for every dart $d$, we can find a minimum
$fg$-separating cycle $C$ in $O(|C|)$ time by the following procedure.
First we can find the node in $\mathcal T$ corresponding to $C$ in
$O(1)$ time using the oracle (Theorem~\ref{Thm:MinCutOracle}).  To
find $C$, walk along $C$, starting with $P(C)$, until you reach the
end.  Let $C_1 = C(d_2(C))$ and walk along $P(C_1)$ starting with
$d_2(C)$.  Suppose we are at dart $d$ along $C$.  Let $C_i = C(d)$.
Walk along $P(C_i)$ until either you reach its end or you hit
dart $d_1(C_{i-1})$.  In the first case, continue the process with
$d_2(C_i)$.  In the second case, continue the process with the first
dart of $P(C_{i-1})$.  By Lemma~\ref{lem:cycle-decomp}, this process
will eventually reach the start of $P(C)$.

\subsection{Preprocessing step}

It remains to show how to precompute $d_1(C)$, $d_2(C)$, $P(C)$ for
every cycle in $\mathcal T$ and $C(d)$ for every dart.  We find these
using the top tree representation of $\mathcal T$ with $O(n\log n)$
preprocessing time.

Let $f_\ell$ and $f_r$ be the faces to the left and right of a dart
$d$.  Then it follows easily from Lemma~\ref{Lem:ExtFaceLCA} and the
clockwise orientation $C(d)$ that $C(d)$ is the bounding cycle of
region $\jump(\lca{f_\ell}{f_r},f_r,1)$ and can be found in
$O(\log n)$ time.

We can easily construct $P(C)$ from the set of darts with
$C(d) = C$.  The ordering can be found just using the endpoints of
these darts so that we can walk along $P(C)$ as required in the
previous section.

To find $d_1(C)$ and $d_2(C)$, we work from leaf to root in $\mathcal
T$ as in the bottom-up algorithm of Section~\ref{sec:SharedEdges}.  We
will show how to find $d_2(C)$.  Finding $d_1(C)$ is symmetric.  For
cycle $C$ we can easily find the last dart $d_\ell$ of $P(C)$.
Consider the darts $d_o$ leaving the endpoint of $P(C)$ in
counterclockwise order, in the embedding, starting with the reverse of
$d_\ell$, we test if $d_o$ is on the boundary of $C$ using
Lemma~\ref{Lem:ExtFaceLCA} in $O(\log n)$ time.  As we test darts we
remove them from further consideration as they will be in the interior
of all ancestor cycles.  In total, this takes $O(n \log n)$ time.

\section{A detailed presentation of planar separators} \label{sec:sep}

In this Section, we show how to find separators satisfying
Definition~\ref{def:sep}, that is, we show how to ensure that Miller's
cycle separator theorem applied to a connected piece yields connected
subpieces.  We also give the proof of Theorem~\ref{thm:sumPiece} which
bounds the size of the pieces and the sum of the squares of their
boundaries in a recursive decomposition.  Although this result has
been required of previous results, specific details, to our knowledge,
is not anywhere else in the literature.  We have included a formal
proof of Theorem~\ref{thm:sumPiece} for completeness.

\subsection{Ensuring Connected Pieces}\label{section-dis}

Now, we show how to ensure that Miller's cycle separator theorem applied
to a connected piece yields connected subpieces. To find a cycle separator
$C$ of the desired size, faces need to be of constant size which we ensure
by triangulating the piece temporarily. Having found $C$, the triangulating
edges are removed and the remaining edges of $C$ induce a separation of the
piece into child pieces. The problem with this approach is that $C$ may cross
a hole multiple times by using the added triangulating edges, giving a
child piece consisting of multiple connected components.

To avoid this problem we triangulate the piece in such a way that $C$ uses
at most two triangulating edges from each hole. We triangulate each hole with
a star centered at a new vertex embedded inside the hole. (This introduces a
high-degree vertex, contradicting our constant-degree assumption, but as we only use
the triangulating edges to find the separator, we still have child
pieces that have constant degree.) The same is done for the external face of the
piece if that face is not a triangle.

Formally, let $P$ be a connected piece with holes and let $P'$ be the
triangulation of $P$ as described above.  Let $C$ be the simple cycle
separator.  Since
$C$ is simple it can use at most two triangulating edges per hole and both
of these edges (if any) are consecutive along $C$.  Let
$R_1$ and $R_2$ be the two closed regions of the plane bounded by $C$.
The child pieces are $P_1 = P \cap R_1$ and $P_2 = P \cap R_2$.  Note
that the child pieces share the edges of $C$ that are in $P$.  The
endpoints of these shared edges are separator vertices and so the size
of the decomposition is still bounded by
Lemmas~\ref{Lem:LevelPieceSize}
and~\ref{Lem:LevelPieceBoundariesSquared} and
Theorem~\ref{thm:sumPiece}.

Since $P$ is connected, the boundary of each hole is a (not
necessarily simple) cycle.  We argue that since a separating cycle $C$
uses at most two triangulating edges from each hole, both incident
to the star center, the child pieces
as defined are connected.  Consider, w.l.o.g., child piece $P_1$.
Suppose $C$ uses triangulating edges $(v_1,c)$ and $(c,v_2)$, where
$c$ denotes the star center, for a hole (or the external face) of $P$ with (connected)
boundary $H$. Let $H_1 = H \cap R_1$.  Since $C$ uses no other
triangulating edges from this hole, $H_1$ connects $v_1$ and $v_2$
and is in $P_1$.  The connectivity of $P_1$ then follows from the connectivity of $P$.

To find a recursive subdivision consisting of connected pieces using the approach in~\cite{INSW11} (which is based
on the approach in~\cite{FR06}), the vertex-weighted variant of Miller's cycle separator theorem is applied and we need to handle
the following three cases when applying the theorem to $P'$:
\begin{enumerate}
\item $P_1$ and $P_2$ should each contain at most a constant fraction of the vertices of $P$,
\item $P_1$ and $P_2$ should each contain at most a constant fraction of the boundary vertices of $P$,
\item $P_1$ and $P_2$ should each contain at most a constant fraction of the holes of $P$.
\end{enumerate}
The first resp.\ second case can be handled by distributing vertex weights evenly on the vertices
resp.\ boundary vertices of $P$ and assigning weight $0$ to newly
introduced 'hole' vertices.
For the third case, we distribute vertex weights evenly on the newly
introduced 'hole' vertices and assign weight
zero to all vertices of $P$. Alternating between these cases will
achieve all three properties within a few separations. (see~\cite{INSW11} and~\cite{FR06} for details.)

Note that the top-level piece $G$ is trivially connected. The above
combined with results from~\cite{INSW11} and from~\cite{FR06} then
imply that we can find a recursive subdivision in $O(n\log n)$ time
where all pieces are connected.

\subsection{Bounds on sizes of pieces and boundaries}

The running time of our algorithm depends on the total size of all the
pieces as well as on the sum of squares of piece boundary sizes. This
is also the case for the algorithm of Fakcharoenphol and Rao. They
make the simplifying assumption that a piece of size $r$ has $O(\sqrt
r)$ boundary vertices. Although their construction ensures that piece
sizes and boundary sizes go down geometrically along any root-to-leaf
path in the recursive subdivision tree, the two quantities need not go
down by the same rate since some applications of Miller's separator
theorem may give more unbalanced splits than others. Thus in their
construction, a piece of size $r$ may have more than $O(\sqrt r)$
boundary vertices. Theorem~\ref{thm:sumPiece} bounds the total size of
pieces as well as the sum of squares of boundary sizes.

We observe that only the new boundary vertices are replicated among
the child pieces and we get:
\begin{eqnarray}
  &\sum_i |P_i| \leq |P| +k \sqrt{|P|}\label{eq:2}\\
  &\sum_i|\partial P_i| \leq |\partial P| + k\sqrt {|P|},
  \label{eq:5}
\end{eqnarray}
where $k$ is some constant that depends on the constant in Miller's
Cycle Separator Theorem.

\begin{lemma}\label{Lem:LevelPieceSize}
  Let ${\mathcal P}_i$ be the set of pieces in level $i$ of a
  recursive subdivision of $G$.  Then $\sum_{P\in{\mathcal P}_i}|P| =
  O(n)$.
\end{lemma}

\begin{proof}
  Let $c_{\min}$ be the constant such that pieces of size at most
  $c_{\min}$ are not subdivided further in a recursive subdivision.
  We may assume that no piece has size less than $\frac 1 2 c_{\min}$.
  Let $L(r)$ denote the total number of vertices (counted with
  multiplicity) in the leaf-pieces of the recursive subdivision of an
  $r$-vertex piece $P$.  We will show that $L(r)\leq c_1r
  - c_2\sqrt r$ for suitable constants $c_1$ and $c_2$.  The lemma
  will follow since the number of vertices (counting multiplicity) in
  any level of the recursive subdivision is dominated by the number of
  vertices in the leaves; this is bounded by $L(n)$ which is $O(n)$.

  We prove that $L(r)\leq c_1r - c_2\sqrt r$ by induction.  In the
  base cases, in which a piece of size $r$ is a leaf of the recursive
  subdivision (so $r \in [\frac 1 2 c_{\min}, c_{\min}]$), $L(r) = r$.  This
  is bounded by $c_1 r -c_2 \sqrt r$ so long as $c_{\min} \leq
  c_1(\frac 1 2 c_{\min})-c_2 \sqrt{c_{\min}}$.  Setting $c_{\min}\geq
  \left( c_2/(\frac 1 2 c_1 - 1) \right)^2$ guarantees this.

  Now consider a piece $P$ of size $r > c_{\min}$.  Assume inductively
  that the claim holds for all values smaller than $r$. Let $r_i$ be
  the size of the $i^{th}$ child of $P$; $P$ has $N = O(1)$ children.
  By the inductive hypothesis we get:
  \begin{equation}
    L(r) = \sum_i L(r_i) \leq \sum_i \left(c_1 r_i - c_2
      \sqrt{r_i}\right) = c_1 \sum_i r_i - c_2 \sum_i \sqrt r_i
    \label{eq:1}
  \end{equation}
  We lower bound $\sum_i \sqrt r_i$ by observing that $\sum_i \sqrt
  r_i$ can only be as small as allowed by $r \leq \sum_i r_i$, $r_i
  \in [0, \frac r 2]$ (see Definition~\ref{def:sep}).
  The minimum value occurs when two of the $r_i$'s are equal to $\frac r 2$ and all
  others are zero, giving:
  \begin{equation}
    \sum_i \sqrt r_i \geq 2\sqrt{\frac r 2} = \sqrt{2r}\label{eq:3}
  \end{equation}
  Combining Equations~(\ref{eq:2}),~(\ref{eq:1}), and~(\ref{eq:3}), we
  get that $$L(r) \leq c_1(r+N\sqrt r)-c_2\sqrt{2r} = c_1 r - c_2(\sqrt
  2 - \frac{c_1}{c_2} N)\sqrt r$$
  This completes the induction for $c_2$ sufficiently larger than $c_1N$.
\end{proof}

\begin{lemma}\label{Lem:LevelPieceBoundariesSquared}
  Let ${\mathcal P}_i$ be the set of pieces in level $i$ of a
  recursive subdivision of $G$.  Then $\sum_{P\in{\mathcal
      P}_i}|\partial P|^2 = O(n)$.
\end{lemma}
\begin{proof}
We may assume that, by adding dummy boundary vertices that do not
contribute to children,
\begin{equation}
  \label{eq:4}
  |\partial P| \geq c \sqrt{|P|}\ \mbox{for every piece $P$,}
\end{equation}
where $c$ is a constant that we will pick below.
We will show that for any piece $P$ with children $P_1,\ldots,P_N$:
\begin{equation}
  \sum_j|\partial P_j|^2\leq|\partial P|^2.\label{eq:SumOfSquares}
\end{equation}
The lemma follows from this because, by summing over all pieces in a level we get,
\[
\sum_{P\in{\mathcal P}_i}|\partial P|^2 \leq \sum_{P\in{\mathcal
    P}_{i-1}}|\partial P|^2 \leq \cdots \leq \sum_{P\in{\mathcal
    P}_0}|\partial P|^2 = |\partial G|^2
\]
Since $G$ has only dummy boundary vertices, $|\partial G|^2 = c
(\sqrt{|G|})^2 = c |G|$, which is $O(n)$, as desired.

We now prove Equation~(\ref{eq:SumOfSquares}). In the next equation, the first and second
inequalities follow from Definition~\ref{def:sep} and
Equation~(\ref{eq:4}), respectively:
\begin{equation}
  \label{eq:6}
  |\partial P_j| \leq \frac 1 2 |\partial P| + \sqrt{|P|} \leq
  \left(\frac 1 2 + \frac 1 c \right) |\partial P|
\end{equation}
\begin{eqnarray*}
  \sum_j|\partial P_j|^2 & \leq &  \left(\frac 1
    2 + \frac 1 c \right) |\partial P|\sum_j |\partial P_j| \qquad \mbox{by
    Equation~(\ref{eq:6})} \\
  & \leq & \left(\frac 1
    2 + \frac 1 c \right) |\partial P| \left(|\partial P| + N
    \sqrt{|P|}\right) \qquad \mbox{by Equation~(\ref{eq:2})} \\
  & \leq & \left(\frac 1
    2 + \frac 1 c \right) |\partial P|\left(|\partial P| +N \frac{|\partial P|}{c}\right)
  \qquad \mbox{by Equation~(\ref{eq:4})} \\
  & = & \left(\frac 1
    2 + \frac 1 c \right)\left(1 + \frac{N}{c}\right)|\partial P|^2
  \\
  & \leq & |\partial P|^2 \qquad \mbox{for constant $c$ sufficiently large.}
\end{eqnarray*}
This completes the proof.
\end{proof}

Since the depth of the recursive subdivision is $O(\log n)$,
Lemmas~\ref{Lem:LevelPieceSize}
and~\ref{Lem:LevelPieceBoundariesSquared} imply
Theorem~\ref{thm:sumPiece}.

\section{Lexicographic-shortest paths} \label{sec:unique-short-paths}

In this section we show how to impose uniqueness of shortest paths by breaking ties in a consistent manner, deterministically.  This will prove:

\begin{theorem}\label{thm:lex}
  The algorithms of Theorems~\ref{Thm:MinCutOracle}
  through~\ref{Thm:CycleBasisExplicit} can be made deterministic with
  only an additional $O(\log^2n)$ factor in the preprocessing time.
\end{theorem}

\noindent Let $w:E\rightarrow\mathbb R$ be the weight function on the
edges of $G$. Index the vertices of $G$ from $1$ to $n$. For a
subgraph $H$, define $I(H)$ as the smallest index of vertices in
$H$. Hartvigsen and Mardon~\cite{HM94} showed that there is another
weight function $\tilde{w}$ on the edges of $G$ such that for any pair
of vertices in $G$, (i) there is a unique shortest path between them
w.r.t.\ $\tilde{w}$ and (ii) this path is also a shortest path w.r.t.\
$w$. Furthermore, for two paths $P$ and $P'$ between the same pair of
vertices in $G$, $\tilde{w}(P) < \tilde{w}(P')$ exactly when one of
the following three conditions is satisfied:
\begin{enumerate}
\item $w(P) < w(P')$.
\item $w(P) = w(P')$ and $|P| < |P'|$.
\item $w(P) = w(P')$, $|P| = |P'|$ and
      $I(P\setminus P') < I(P'\setminus P)$.
\end{enumerate}
A shortest path w.r.t.\ $\tilde{w}$ is called a \emph{lex-shortest path} and a shortest path tree w.r.t.\ $\tilde{w}$ is called
a \emph{lex-shortest path tree}.
The properties of $\tilde{w}$ allow us to use $\tilde{w}$ instead
of $w$ in our algorithm. In the following, we show how to do so
efficiently.

We first use a small trick from Hartvigsen and Mardon~\cite{HM94}: for
function $w$, we add a sufficiently small $\epsilon > 0$ to the weight
of every edge. This allows us to disregard the second condition
above. When comparing weights of paths, we may treat $\epsilon$
symbolically so we do not need to worry about precision issues. The
tricky part is efficiently testing the third condition.

We need to make modifications to every part of our algorithm in which
the weights of two shortest paths are compared. All such comparisons
occur when we (1)~use Fakcharonphol and Rao's variant of Dijkstra's
algorithm, FR-Dijkstra \cite{FR06} and (2)~find a shortest path
covering of an isometric cycle $C$ in
Section~\ref{subsubsec:SPCovering}.

\subsection{FR-Dijkstra}
Let us first adapt FR-Dijkstra to compute lex-shortest paths. The type
of shortest path weight comparisons in FR-Dijkstra are of the form
$D(u) + d(u,v) < D(u') + d(u',v)$, where $u$, $v$, $u'$, and $v'$ are
vertices, $D(u)$ and $D(u')$ are the distances from the root of the
partially built tree to $u$ and $u'$, respectively, and $d(u,v)$ and
$d(u',v)$ are the lengths or weights of edges $(u,v)$ and
$(u',v)$. Note that an edge can be an edge of $G$ (in which case
$d(u,v) = w(u,v)$) or be a cycle edge (Section~\ref{sec:SharedEdges})
or an edge of an external or internal dense distance graph (in which
case $d(u,v)$ is the length of the path the edge represents).

For simplicity, assume first that all edges considered by FR-Dijkstra
belong to $G$; we test whether $D(u) + d(u,v) < D(u') + d(u',v)$ as
follows. Let $T$ be the partially built shortest path tree rooted at a
vertex $r$ and let $Q$ and $Q'$ be the $r$-to-$u$ and $u'$ paths in
$T$, respectively. If the first two lex-shortest conditions are
inconclusive, we need to check if $I(Q\setminus Q') < I(Q'\setminus
Q)$.

Let $a$ be the least-common ancestor of $u$ and $u'$ in $T$.  Then $Q
\setminus Q'$ is the $a$-to-$u$ subpath of $Q$, excluding $a$. It
follows from this that, by representing $T$ as a top tree, we can find
the smallest index in the two sets in logarithmic time.  Using top
trees, we can also similarly handle a cycle edge $e$, by keeping the
smallest index of $e$'s interior vertices. These indices can be found
during the construction of region subpieces in
Section~\ref{sec:ident-region} without an increase in running time.

\subsection{Internal dense distances}\label{sec:FRint}
We also need to handle edges from internal and external dense distance graphs. Let us first consider the problem of computing
lex-shortest path trees in $\intDDG$. As before, we compute shortest path trees for pieces bottom-up. Let $P$ be a piece with
children $P_1$ and $P_2$ and assume that we have computed lex-shortest path trees in both of them. Assume also that every
edge in $\intDDG(P_1)\cup\intDDG(P_2)$ is associated with the smallest index of interior vertices on the path in $G$ that the
edge represents. This information can be computed bottom-up during the construction of $\intDDG$ without increasing running time.

Let $T$ be a partially-built shortest-path tree in $P$.  With the
above definitions, consider the problem of testing whether $D(u) +
d(u,v) < D(u') + d(u',v)$. Let $(a,u_a)$ and $(a,u_a')$ be the first
edges on $Q[a,u]$ and $Q'[a,u']$, respectively, with $a$ defined as
earlier.  Define $Q_G$ and $Q_G'$ as the paths in $G$ represented by
$Q$ and $Q'$, respectively. Let $i_1 = I(Q_G[u_a,u])$, $i_1' =
I(Q_G'[u_a',u'])$, $i_2 = I(Q_G[a,u_a]\setminus Q_G'[a,u_a'])$, and
$i_2' = I(Q_G'[a,u_a']\setminus Q_G[a,u_a])$. By definition of $a$,
$Q_G[u_a,u]$ and $Q_G'[u_a',u']$ are vertex-disjoint. We need to
compute these indices and check if $\min\{i_1,i_2\} <
\min\{i_1',i_2'\}$.

Each edge of $T$ belongs to $\intDDG(P_1)\cup\intDDG(P_2)$ and is thus associated with the smallest index of interior vertices
on the path in $G$ represented by the edge. Top tree operations on $T$ as above then allow us to find $i_1$ and $i_1'$ in
logarithmic time.

To find $i_2$ and $i_2'$, we consider two cases: $(a,u_a)$ and
$(a,u_a')$ belong to the internal dense distance graph for the same
child of $P$ or they belong to different graphs. In the first case, assume that, say, $(a,u_a),(a,u_a')\in\intDDG(P_1)$.
Then we can decompose these two edges into shortest paths in the same shortest path tree in $\intDDG(P_1)$ and we can recursively
find $i_2$ and $i_2'$. In the second case, assume that, say, $(a,u_a)\in\intDDG(P_1)$ and $(a,u_a')\in\intDDG(P_2)$. Since the
lex-shortest paths representing these edges in $\intDDG(P_1)$ and $\intDDG(P_2)$ are edge-disjoint and since $T$ is a partially
built lex-shortest path tree in $P$, $Q_G[a,u_a]$ and $Q_G'[a,u_a']$  share no vertices except $a$. Thus, $i_2$
is the smallest index of vertices in $V(Q_G[a,u_a])\setminus\{a\}$ and we can obtain this index in constant time from the index of
$u_a$ and the index associated with edge $(a,u_a)$ which is the smallest index of interior vertices on $Q_G[a,u_a]$. Similarly,
we can find $i_2'$ in constant time.

Since the subdivision tree has $O(\log n)$ height, the recursion depth of the above algorithm is $O(\log n)$, implying that we
can determine whether $D(u) + d(u,v) < D(u') + d(u',v)$ in $O(\log^2n)$ time. Hence, lex-shortest path trees in $\intDDG$ can
be computed in a total of $O(n\log^4n)$ time.

\subsection{External dense distances}
Computing lex-shortest path trees in $\extDDG$ within the same time bound is very similar so we only highlight the differences.
Having computed lex-shortest path trees in $\intDDG$ bottom-up, we compute lex-shortest path trees in $\extDDG$ top-down.
For a piece $P$, we obtain lex-shortest path trees from lex-shortest path trees in its sibling and parent pieces. We can then
use an algorithm similar to the one above to find lex-shortest path trees in $P$. At each recursive step, we either go up
one level in $\extDDG$ or go to $\intDDG$. It follows that the recursive depth is still $O(\log n)$ so
lex-shortest path trees in $\extDDG$ can be found in $O(n\log^4n)$ time.

\subsection{FR-Dijkstra in Reif's algorithm}
We also use FR-Dijkstra in Section~\ref{sec:SepFacePair} to emulate
Reif's minimum separating cycle algorithm.  First, we computed a
shortest path $X$ between two faces of the region subpiece using
FR-Dijkstra. With an algorithm similar to the one above, we can
instead compute a lex-shortest path between the two faces with an
$O(\log^2n)$ time overhead.  Next, we cut open the region subpiece
along this path.  The handling of external distances in the cut-open
graph does not change, but the internal distances are recomputed.  We
recompute these as in Section~\ref{sec:FRint}.

\subsection{Shortest path coverings}
In Section~\ref{subsubsec:SPCovering}, we gave an algorithm to find the unique edge $e = (u,v)$ on isometric cycle $C$ such that
the two shortest paths from a fixed vertex $r$ on $C$ to $u$ and to $v$ cover all vertices of $C$ and all edges except $e$.
We showed how to do this in $O(|C| + \log^3n)$ time, where $|C|$ is
the size of the compact representation of $C$ obtained in
Section~\ref{sec:SepFacePair}. We need to modify the algorithm to do so with respect to lex-shortest paths.

Recall that to find $e$, a linear search of the super edges of $C$
from $r$ was first applied to find the super edge $\hat e$ of $C$ such
that the shortest path in $G$ representing $\hat e$ contains $e$. As
above, we may assume that every super edge of $C$ is associated with
the smallest index of interior vertices on the path it
represents. Hence, by keeping track of the smallest interior vertex
index for super edges visited so far in the linear search as well as the
smallest interior vertex index for edges yet to be visited, we can
find $\hat e$ in $O(|C|)$ time w.r.t.\ lex-shortest paths.

Having found $\hat e$, we need to apply binary search on a path
representing $\hat e$ in a lex-shortest path tree. We do this by
first finding the midpoint of this path as in Section~\ref{subsubsec:SPCovering}. If the two halves have the same weight
and the same number of edges, we can use a top tree operation on each half to determine which half has the smallest index.
It follows that all binary searches to find $e$ take $O(\log^3n)$ time. The total time to find $e$ is thus $O(|C| + \log^3n)$,
which matches the time in Section~\ref{subsubsec:SPCovering}.

\bibliographystyle{plain}
\bibliography{all-gh}

\end{document}